\documentclass{wvdareportstyle}
\usepackage{url} 
\usepackage[ruled,lined]{algorithm2e}
\usepackage{graphicx}

\usepackage{amssymb}
\usepackage{amsmath}
\usepackage{times}
\usepackage{enumitem}
\usepackage{mymacros}
\usepackage{eurosym}

\def\univs{\mathbb{U}}
\newcommand{\univ}[1]{\univs_{\mathit{#1}}}
\newcommand{\pim}[1]{\pi_{\mathit{#1}}}

\usepackage{listings}
\usepackage{booktabs}
\usepackage{pgfplots}

\begin{document}
\setcounter{page}{1}

\title{Discovering Object-Centric Petri Nets}

\author{Wil M.P. van der Aalst \and Alessandro Berti\\
Process and Data Science (PADS), RWTH Aachen University, Aachen, Germany \\
Fraunhofer Institute for Applied Information Technology, Sankt Augustin, Germany\\
wvdaalst,a.berti{@}pads.rwth-aachen.de}

\runninghead{W.M.P. van der Aalst and A. Berti}{Discovering Object-Centric Petri Nets}
\maketitle

\begin{abstract}
Techniques to discover Petri nets from event data assume precisely one case identifier per event.
These case identifiers are used to correlate events, and the resulting discovered Petri net aims to describe the life-cycle of individual cases.
In reality, there is not one possible case notion, but multiple intertwined case notions.
For example, events may refer to mixtures of orders, items, packages, customers, and products.
A package may refer to multiple items, multiple products, one order, and one customer.
Therefore, we need to assume that each event refers to a collection of objects, each having a type (instead of a single case identifier).
Such \emph{object-centric event logs} are closer to data in real-life information systems.
From an object-centric event log, we want to discover an \emph{object-centric Petri net} with places that correspond to object types and
transitions that may consume and produce collections of objects of different types.
Object-centric Petri nets visualize the complex relationships among objects from different types.
This paper discusses a novel process discovery approach implemented in PM4Py.
As will be demonstrated, it is indeed feasible to discover holistic process models
that can be used to drill-down into specific viewpoints if needed.
\end{abstract}


\section{Introduction}
\label{sec:intro}

The synthesis of ``higher-level'' process models from ``lower-level'' behavioral specifications has been subject of active research for decades.
Examples of such ``higher-level'' process models are (colored) Petri nets, BPMN models, Statecharts, etc.
Examples of ``lower-level'' behavioral specifications serving as input for synthesis are transition systems, languages, partial orders, and scenarios.
In the context of Petri nets, the ``Theory of Regions'' has been very influential.
Regions were introduced for elementary nets and transition systems in the seminal paper \cite{ehrenfeucht_regions}.
The goal was to create a Petri net with a reachability graph that is isomorphic to the transition system used as input. The core idea has been generalized in numerous directions.
Different classes of target models have been investigated \cite{Badouel-Darondeau-book-regions-2015,BaDa98,DeRe96,Jetty-TCS-2017}, e.g., bisimilar Place Transition (P/T) nets \cite{Cortadella98}, Petri nets with arc weights \cite{DBLP:conf/apn/CarmonaCKKLY08,DBLP:journals/tc/CarmonaCK10}, Petri nets with  a/sync connections \cite{Jetty-TCS-2012}, $\tau$-nets \cite{Jetty-TCS-2017}, zero-safe nets \cite{zero-safe-nets-synthesis}, etc.
Typically, a transition system is used as input.
However, there are various region-based approaches taking as input
languages \cite{lorenz_BPM2007,Lorenz_ACSD_07,DBLP:journals/fuin/BergenthumDLM08,boudewijn_runs_ToPNoC-special-issue-PN2011,DBLP:journals/fuin/BergenthumDML09,Lorenz_WSC_07}, partial orders/scenarios \cite{lorenz_atpn_2006,DBLP:conf/apn/BergenthumDLM08,models-from-scenarios-ToPNoC}, or other ``lower-level'' behavioral specifications.

\emph{Process mining} is related to the field of synthesis (in particular language-based regions). 
However, the assumptions and goals are very different. 
Whereas classical synthesis approaches aim to obtain a
``higher-level'' process model that compactly describes the behavior of a ``lower-level'' behavioral specification, process mining techniques face a more difficult problem.
The event logs used as input for process discovery typically contain only a fraction of the possible behavior. Traces in an event log can be seen as examples. If there are loops,
one cannot expect to see all possible traces.
If a model contains concurrency, one cannot expect to see all possible interleavings.
If the model has multiple choices, one cannot expect to witness all possible combinations.
There have been many attempts to extend region-based approaches to this setting \cite{DBLP:conf/bpm/CarmonaCK08,lorenz_BPM2007,carmona-PN2010,bas-ilp-computing}.
Unfortunately, region-based techniques are often computationally intractable, lead to overfitting models, and/or cannot discover process constructs such as skipping and mixtures of choice and synchronization (e.g., OR-joins).
Hence, several more scalable and robust techniques have been developed.
Commercial tools typically still resort to learning the so-called \emph{Directly Follows Graph} (DFG) which typically leads to underfitting process models \cite{centeris-keynote2019}. When activities appear out of sequence, loops are created, thus leading to Spaghetti-like diagrams suggesting repetitions that are not supported by the data.
The inductive mining techniques \cite{sander-infreq-bpi2013-lnbip2014,sander-scalable-BPMDS2015} and the so-called split miner \cite{split-miner} are examples of the state-of-the-art techniques to learn process models.
These techniques are able to generalize and uncover concurrency.
\begin{figure}[htb!]
\centerline{\includegraphics[width=13cm]{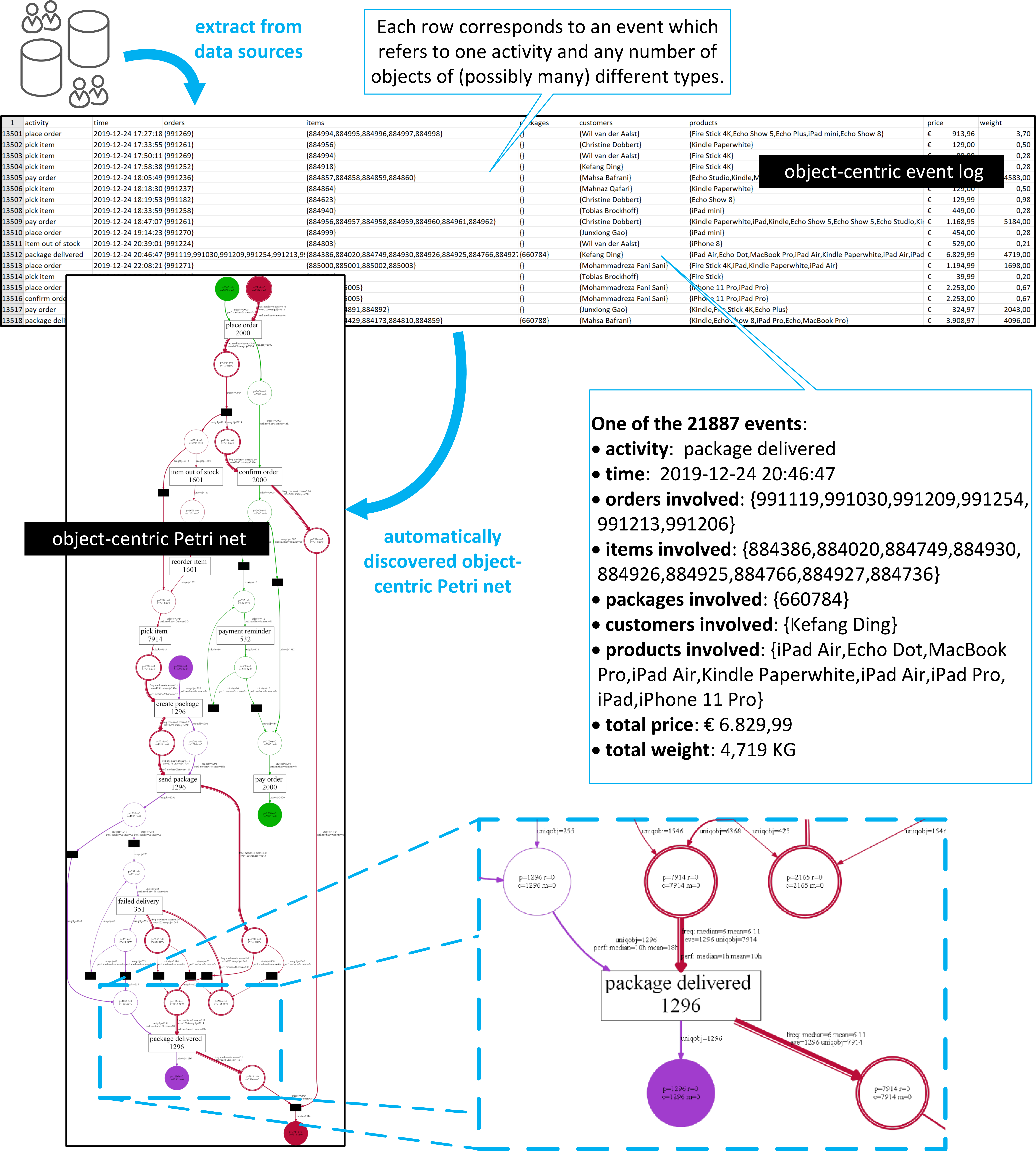}}
\caption{Overview of the approach presented in this paper. Object-centric event logs are used as an intermediate format in between the actual data sources and traditional event logs requiring a single case notion. Using this input, we discover object-centric Petri nets that are able to describe multiple object types in a single model.}
\label{f-intro}
\end{figure}

This paper focuses on process discovery. However, rather than presenting a new discovery technique for traditional event logs, we start from \emph{object-centric event logs} \cite{wvda-keynote-SEFM2019}.
In a traditional event log each event is related to one activity, one timestamp, and one case (i.e., a process instance). We still make the assumption that each event refers to an activity and a point in time.
However, we do \emph{not} assume the existence of a single case notion. \emph{Instead, an event may refer to any number of objects and these objects may be of different types.}
This extends the reach of process mining dramatically. The step is comparable to going from Place Transition (P/T) nets to colored Petri nets.
\emph{Objects can be viewed as colored tokens and object types can be seen as color sets (i.e.,  place types).}

\emph{Based on object-centric event logs, we aim to automatically discover object-centric Petri nets}.
Such Petri nets have typed places that refer to the object types in the event log.
Just like in colored Petri nets, a transition may consume or produce \emph{multiple} tokens from a place during one execution.
In this paper, we present the first technique to discover such nets.
In the related work section (Section~\ref{sec:rw}), we elaborate on the relation to earlier approaches such as the Object-Centric Behavioral Constraint (OCBC) models \cite{BIS-OCBCdisc-lnbip2017},
synchronized transitions systems \cite{DBLP:conf/wecwis/EckSA17,Multi-instance-Mining-BPM-WS-2018}, and artifact-centric discovery approaches \cite{BIS-artifactconformance-lnbip2011,zeus-artifact-2011,Xixi-TSC-2015}.

Figure~\ref{f-intro} illustrates the approach presented.
Object-centric event logs can be extracted from any information system \cite{wvda-keynote-SEFM2019}. These logs can be seen as an intermediate format closer to the actual data collected by today's information systems.
Unlike traditional event logs (e.g., XES logs), an event may refer to multiple objects and is not forced to be assigned to a single case.
Enterprise Information Systems (EIS), Customer Relationship Management (CRM) systems, Healthcare Information Systems (HIS), E-Learning Systems, Production Systems, Supply Chain Systems, etc.\
typically store information on a range of objects (customer, orders, patients, products, payments, etc.) in multiple tables that refer to each other.
Figure~\ref{f-intro} shows the objects related to one ``package delivered'' event. The event refers to six orders, nine items, one package, one customer, and nine products.
In total, there are 22,367 events. Figure~\ref{f-intro} shows an object-centric Petri net (not intended to be readable) discovered while focusing on orders, items, and packages.
The places and arcs are typed. The colors red, green, and purple refer to respectively orders, items, and packages.

Just like for traditional process mining approaches it is possible to filter and seamlessly simplify the process model. By focusing on a particular object type, it is also
possible to create traditional events logs that can be analyzed using traditional process mining techniques.

The remainder of this paper is organized as follows.
Section~\ref{sec:prelim} introduces event logs and process models.
Object-centric event logs are introduced and motivated in Section~\ref{sec:ocelogs}.
Given such logs, we first discuss techniques to learn process models for a single object type in Section~\ref{sec:discot}.
In Section~\ref{sec:ocpn}, we introduce object-centric Petri nets, i.e., Petri nets with places referring to object types.
Section~\ref{sec:discoopn} presents the main contribution of this paper: An approach to learn object-centric Petri nets from object-centric event logs.
The discovery technique has been implemented in \emph{PM4Py}, an open-source process mining platform written in Python.
Section~\ref{sec:impl} presents the implementation and Section~\ref{sec:appl} demonstrates the feasibility of the approach.
Related work is discussed in Section~\ref{sec:rw}.
Section~\ref{sec:concl} concludes the paper with a few final remarks.

\section{Preliminaries}
\label{sec:prelim}

First, we introduce some preliminaries for people not familiar with process mining and accepting Petri nets.
Input for process mining is an event log.
A \emph{traditional} event log views a process from a particular angle provided by the \emph{case notion} that is used to \emph{correlate events}.
Each event in such an event log refers to (1) a particular \emph{process instance} (called \emph{case}), (2) an \emph{activity}, and (3) a \emph{timestamp}.
There may be additional event attributes referring to resources, people, costs, etc., but these are optional.
With some effort, such data can be extracted from any information system supporting operational processes.
Process mining uses these event data to answer a variety of process-related questions.
Process mining techniques such as process discovery, conformance checking, model enhancement, and operational support can be used to improve performance and compliance \cite{process-mining-book-2016}.

Each event in an event log has three \emph{mandatory} attributes: 
case, activity, and timestamp.
The case notion is used to group events, e.g., all events corresponding to the same order number are taken together. The timestamps are used to order the events and can be used to analyze bottlenecks, delays, etc. There may be many additional attributes, e.g., costs, resource, location, etc.
However, most process discovery techniques first learn a model where 
only the order of activities within cases matters.
Once the control-flow is clear, other attributes (e.g., time) can be added by replaying the event log on the model \cite{process-mining-book-2016}.
Therefore, we define a so-called ``simple event log'' that only records the ordering of activities for each case.
Technically, an event log is a multiset of traces. $B \in \bag(X) = X \rightarrow \Nat$ is a multiset over $X$ where element $x \in X$ appears $B(x)$ times. For example, in $B=[a^5,b^2,c]$, $a$ appears $B(a)=5$ times, $b$ twice, and $c$ once.

\begin{definition}[Simple Event Log]\label{def:sel}
Let $\univ{act}$ be the universe of activity names.
A trace $\sigma \in \univ{act}^*$ is a sequence of activities.
$L\in \bag(\univ{act}^*)$ is an event log, i.e., a multiset of traces.
$\univ{SEL} = \bag(\univ{act}^*)$ is the universe of simple event logs.
\end{definition}

For example, 
$\univ{act} = \{ \mi{po}, \mi{pi}, \mi{sh}, \mi{in}, \mi{sr}, \mi{pa}, \mi{co}, \ldots  \}$ where
$\mi{po}$ denotes activity \emph{place order},
$\mi{pi}$ denotes activity \emph{pick item},
$\mi{sh}$ denotes activity \emph{ship item},
$\mi{in}$ denotes activity \emph{send invoice},
$\mi{sr}$ denotes activity \emph{send reminder},
$\mi{pa}$ denotes activity \emph{pay order}, and
$\mi{co}$ denotes activity \emph{mark as completed}.
Using this more compact notation we show three example traces:
$\sigma_1 = \langle \mi{po},\allowbreak \mi{in},\allowbreak \mi{pi},\allowbreak \mi{sr},\allowbreak \mi{sh},\allowbreak \mi{pa},\allowbreak \mi{co} \rangle$,
$\sigma_2 = \langle \mi{po},\allowbreak \mi{pi},\allowbreak\mi{sh},\allowbreak \mi{in},\allowbreak \mi{pa},\allowbreak \mi{co} \rangle$, and
$\sigma_3 = \langle \mi{po},\allowbreak \mi{in},\allowbreak \mi{sr},\allowbreak \mi{sr},\allowbreak \mi{pi},\allowbreak \mi{sr},\allowbreak \mi{pa},\allowbreak \mi{sh},\allowbreak \mi{co} \rangle$.
Obviously, multiple cases can have the same trace.
In an event log $L= [ {\sigma_1}^{435}, {\sigma_2}^{366}, {\sigma_3}^{233}, \ldots  ]$ the above three traces appear respectively 435, 366, and 233 times.
Given such an event log, process discovery techniques are able to learn 
a process model describing the observed traces.
Such techniques often take into account frequencies, e.g., the model should cover the most frequent traces but may decide to abstract from infrequent ones.
Figure~\ref{f-example-pm-1} shows a process model discovered for event log $L$.
\begin{figure}[thb]
\centerline{\includegraphics[width=7cm]{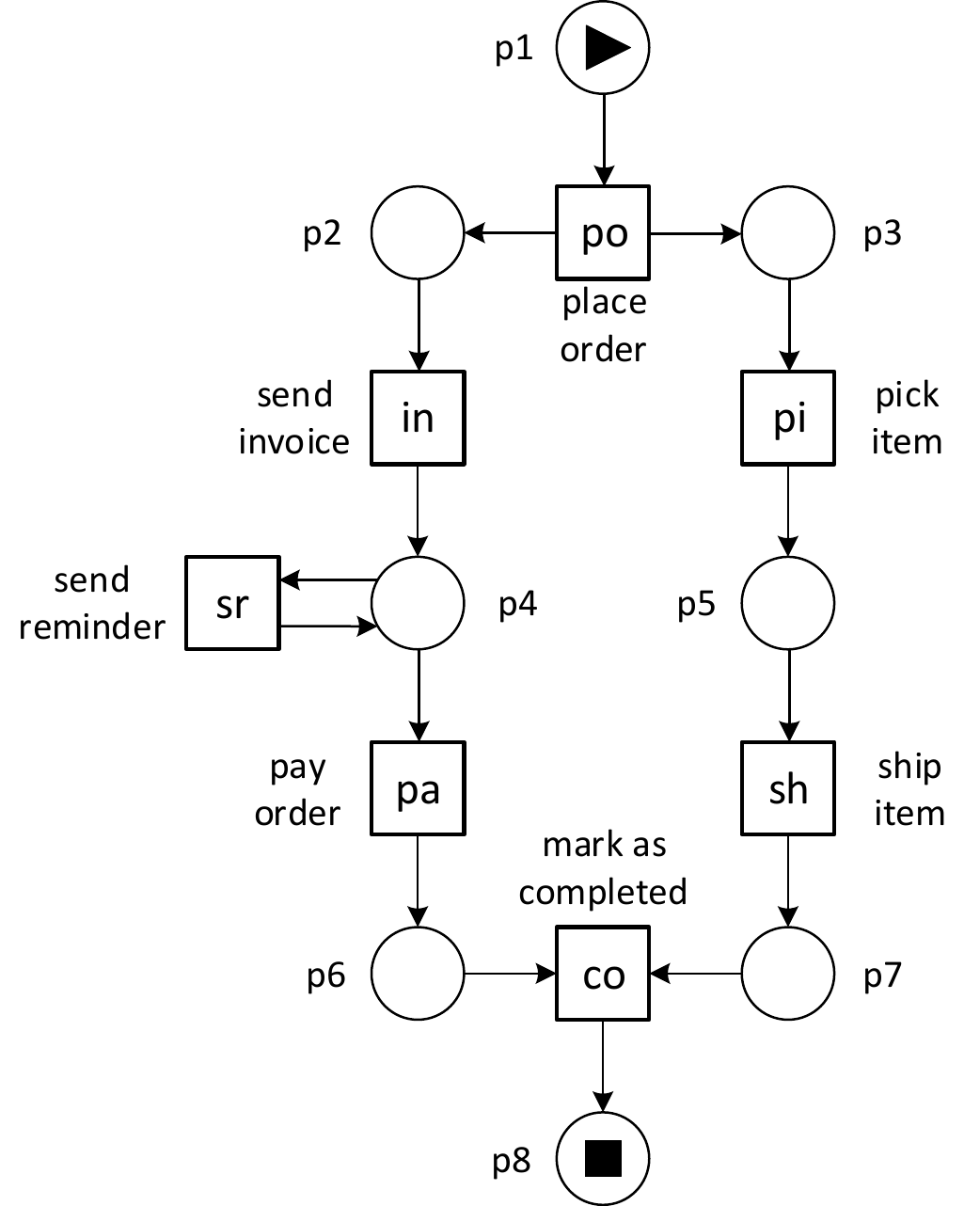}}
\caption{An accepting Petri net composed of eight places and seven transitions.}
\label{f-example-pm-1}
\end{figure}

The discovered process model in Figure~\ref{f-example-pm-1} is represented as an \emph{accepting} Petri net where the transitions are \emph{labeled}.
We assume that the reader is familiar with standard Petri nets notations, but provide a few definitions to make the key notions explicit. 
We use Petri nets with a \emph{labeling function} and \emph{final marking}.
This is driven by requirements from process mining. The labeling function is needed to model skips and duplicate activities. The final marking is needed because traces 
have a defined start and end.

\begin{definition}[Labeled Petri Net]\label{def:lpn}
A labeled Petri net is a tuple $N=(P,T,F,l)$ with $P$ the set of places, $T$ the set of transitions,
$P \cap T = \emptyset$, $F\subseteq (P \times T) \cup (T \times P)$ the flow relation, and $l \in T \not\rightarrow \univ{act}$  a labeling function.
\end{definition}

A Petri net defines a directed graph with nodes $P\cup T$ and edges $F$.
The state of a Petri net, called \emph{marking}, is a multiset of places ($M \in \bag(P)$).
A transition $t \in T$ is \emph{enabled} in marking $M$ of net $N$ if each of its input places $\pre t = \{p \in P\mid (p,t) \in F\}$ contains at least one token.
An enabled transition $t$ may \emph{fire}, i.e., one token is removed from each of the input places $\pre t$ and
one token is produced for each of the output places $\post t = \{p \in P\mid (t,p) \in F\}$. Assume that $[\mi{p1}]$ is the initial marking of the Petri net in Figure~\ref{f-example-pm-1}.
There are 11 markings reachable from this initial marking, including $[\mi{p1}]$, $[\mi{p2},\mi{p3}]$, $[\mi{p4},\mi{p7}]$, and $[\mi{p8}]$.

Note that the labeling function $l$ may be partial and non-injective.
This means that multiple transitions may refer to the same activity and that there may be transitions that are ``silent'' and do not correspond to an activity.
Any firing sequence of a labeled Petri net corresponds to a visible trace obtained by mapping transitions onto activities using $l$.
Firing an unlabeled transition does not add an activity to the trace. In Figure~\ref{f-example-pm-1}, all transitions are visible and unique.
$\sigma_1$, $\sigma_2$, and $\sigma_3$ are examples of visible traces (assuming the short names as activity labels).

For process mining, we often focus on so-called \emph{accepting Petri nets} that have an initial marking and a final marking.
The reason is that we want to have a model that defines a language corresponding to the process that was used to produce the event log.

\begin{definition}[Accepting Petri Net]\label{def:apn}
An accepting Petri net is a triplet $\mi{SN}=(N,M_{\mi{init}},M_{\mi{final}})$ where
$N=(P,T,F,l)$ is a labeled Petri net,
$M_{\mi{init}} \in \bag(P)$ is the initial marking, and $M_{\mi{final}} \in \bag(P)$ is the final marking.
$\univ{APN}$ is the universe of accepting Petri nets.
\end{definition}

In Figure~\ref{f-example-pm-1}, the initial marking  $M_{\mi{init}} = [\mi{p1}]$ and the final marking $M_{\mi{final}} = [\mi{p8}]$ are denoted using the start and stop symbol.

\begin{definition}[Language of an Accepting Petri Net]\label{def:lang}
An accepting Petri net $\mi{SN}=(N,M_{\mi{init}},M_{\mi{final}})$ defines a language
$\phi(\mi{SN})$ that is composed of all \emph{visible} traces (ignoring transition occurrences not having a label) starting in $M_{\mi{init}}$ and ending in $M_{\mi{final}}$.
\end{definition}

The accepting Petri net depicted in Figure~\ref{f-example-pm-1} has infinitely many visible traces due to the loop involving $\mi{sr}$.
Without the loop, there would be six possible visible traces.

Assuming the basic setting with simple event logs and accepting Petri nets, we can now formally define the notion of \emph{process discovery}.
For any event log, we would like to construct a corresponding process model.

\begin{definition}[Process Discovery Technique]\label{def:pd}
Discovery technique $\mi{disc}$ is a function mapping simple event logs onto accepting Petri nets, i.e., $\mi{disc} \in \univ{SEL} \rightarrow \univ{APN}$.
\end{definition}

What makes process mining very difficult is that the event log only contains example behaviors.
If Figure~\ref{f-example-pm-1} represents the real process, we have the problem that no event log will contain all of its traces (due to the loop).
Even when there are no loops, it is very unlikely to observe all possible traces for real-life processes due to combinations of choices and the interleaving of concurrent activities.
Typically, only a fraction of the possible process is observed.
Moreover, the event log may contain noise and infrequent behaviors that should not end up in the process model.
This leads to notions such as recall (also called fitness), precision, generalization, and simplicity \cite{process-mining-book-2016}.
These are outside of the scope of this paper. However, we abstractly define the notion of \emph{conformance checking}.

\begin{definition}[Conformance Checking Technique]\label{def:cc}
Conformance checking technique $\mi{conf}$ is a function mapping a pair composed of an event log and an accepting Petri nets onto conformance diagnostics, i.e., $\mi{conf} \in (\univ{SEL} \times \univ{APN}) \rightarrow \univ{diag}$.
\end{definition}

$\mi{conf}(L,\mi{SN}) \in \univ{diag}$ provides diagnostics related to recall, precision, generalization, simplicity, etc.
An example would be the fraction of traces in the event log that can be replayed by the accepting Petri net: $\mi{conf}(L,\mi{SN}) = \card{[\sigma \in L \mid \sigma \in \phi(\mi{SN})]}/\card{L}$.
Given $L' = [
\langle \mi{po},\allowbreak \mi{pi},\allowbreak\mi{sh},\allowbreak \mi{in},\allowbreak \mi{pa},\allowbreak \mi{co} \rangle^8, \langle \mi{po},\allowbreak\mi{sh},\allowbreak \mi{pi},\allowbreak \mi{in},\allowbreak \mi{pa},\allowbreak \mi{co} \rangle^2]$ and
$\mi{SN}$ shown in Figure~\ref{f-example-pm-1}, $\mi{conf}(L',\mi{SN}) = 0.8$ given this conformance notion.
Many other measures and diagnostics are possible. However, we leave $\univ{diag}$ deliberately vague.

\section{Object-Centric Event Logs}
\label{sec:ocelogs}

Section~\ref{sec:prelim} provided a basic introduction to process mining, assuming that there is a clear case notion.
In this section, we show that, for many applications, this assumption is not realistic (Section~\ref{subsec:nocase}).
Next, we formalize the notion of \emph{object-centric event logs} (Section~\ref{subsec:formallog}).

\subsection{What If There is not a Single Case Identifier?}
\label{subsec:nocase}

In many applications, there are multiple candidate case notions leading to different views on the same process \cite{wvda-keynote-SEFM2019}.
Moreover, one event may be related to different cases (\emph{convergence}) and, for a given case, there may be multiple instances of the same activity within a case (\emph{divergence}).
To create a traditional process model, the event data need to be ``flattened''. There are typically multiple choices possible, leading to different views that are disconnected or inconsistent.

To introduce the problem, consider the event log shown in Table~\ref{tablogtwoobjecttypes}.
The table shows that each order may correspond to multiple items 
that are picked and shipped separately.
This is a more realistic assumption (shops tend to allow customers to buy more than one product per order).
\begin{table}[htb!]
\caption{A fragment of an event log: Each line corresponds to an event, possibly referring to multiple objects (i.e., orders and items).}\label{tablogtwoobjecttypes}
\centering
\resizebox{0.75\columnwidth}{!}{\begin{tabular}{|c|c|c|c|}
\hline
activity & timestamp & order & item \\ \hline
 $\ldots$ &  $\ldots$ &  $\ldots$ &  $\ldots$\\
\emph{place order} & 25-11-2019:09.35 & $\{99001\}$ & $\{88124,88125,88126\}$\\
\emph{pick item} & 25-11-2019:10.35 & $\emptyset$ & $\{88126\}$ \\
\emph{place order} & 25-11-2019:11.35 & $\{99002\}$ & $\{88127,88128\}$\\
\emph{pick item} & 26-11-2019:010.25 & $\emptyset$ & $\{88124\}$ \\
\emph{send invoice} & 27-11-2019:08.12 & $\{99001\}$ & $\emptyset$\\
\emph{send invoice} & 28-11-2019:09.35 & $\{99002\}$ & $\emptyset$\\
\emph{pick item} & 29-11-2019:09.35 & $\emptyset$ & $\{88127\}$ \\
\emph{send reminder} & 29-11-2019:10.35 & $\{99002\}$ & $\emptyset$\\
\emph{pick item} & 29-11-2019:11.15 & $\emptyset$ & $\{88128\}$ \\
\emph{ship item} & 29-11-2019:12.35 & $\emptyset$ & $\{88124\}$ \\
\emph{pick item} & 29-11-2019:13.30 & $\emptyset$ & $\{88125\}$ \\
\emph{send reminder} & 29-11-2019:14.35 & $\{99001\}$ & $\emptyset$\\
\emph{ship item} & 29-11-2019:15.15 & $\emptyset$ & $\{88125\}$ \\
\emph{send reminder} & 29-11-2019:16.15 & $\{99002\}$ & $\emptyset$\\
\emph{ship item} & 29-11-2019:17.45 & $\emptyset$ & $\{88126\}$ \\
\emph{ship item} & 29-11-2019:18.00 & $\emptyset$ & $\{88128\}$ \\
\emph{send reminder} & 30-11-2019:09.35 & $\{99002\}$ & $\emptyset$\\
\emph{ship item} & 30-11-2019:10.05 & $\emptyset$ & $\{88127\}$ \\
\emph{pay order} & 30-11-2019:11.45 & $\{99002\}$ & $\emptyset$\\
\emph{pay order} & 30-11-2019:12.55 & $\{99001\}$ & $\emptyset$\\
\emph{mark as completed} & 01-12-2019:09.35 & $\{99001\}$ & $\{88124,88125,88126\}$\\
\emph{place order} & 02-12-2019:10.40 & $\{99003\}$ & $\{88129\}$\\
\emph{mark as completed} & 04-12-2019:11.05 & $\{99002\}$ & $\{88127,88128\}$\\
\emph{place order} & 06-12-2019:14.18 & $\{99004\}$ & $\{88130,88131,88132,88133,88134\}$\\
 $\ldots$ &  $\ldots$ &  $\ldots$ & $\ldots$\\
 \hline
\end{tabular}}
\end{table}

Table~\ref{tablogtwoobjecttypes} has a column for order identifiers and item identifiers.
Order 99001 corresponds to three items (88124, 88125, and 88126),
order 99002 corresponds to two items (88127 and 88128),
order 99003 corresponds to one item (88129),  and
order 99004 corresponds to five items (88130, 88131, 88132, 88133, and 88134).
The pick and ship activities are executed for individual items.
An order is marked as completed when all items have been picked and shipped and the order itself was paid.
Note that the events \emph{place order} and \emph{mark as completed} for order 99001, both refer to four objects (one order and three items).
The latter number is variable. For example, the event \emph{place order} for order 99003 refers to only two objects.
This cannot be expressed using the accepting Petri nets introduced before. Transitions need to consume and produce a variable number of tokens of different types.
Therefore, we propose to use \emph{object-centric Petri nets}. Note that we do \emph{not} propose such nets as a new \emph{modeling} language. It can be viewed as a subclass of colored Petri nets, but our focus is on learning a model describing the data in Table~\ref{tablogtwoobjecttypes}.
Hence, we limit the modeling notation to what can be discovered for such data.
\begin{figure}[htb]
\centerline{\includegraphics[width=7cm]{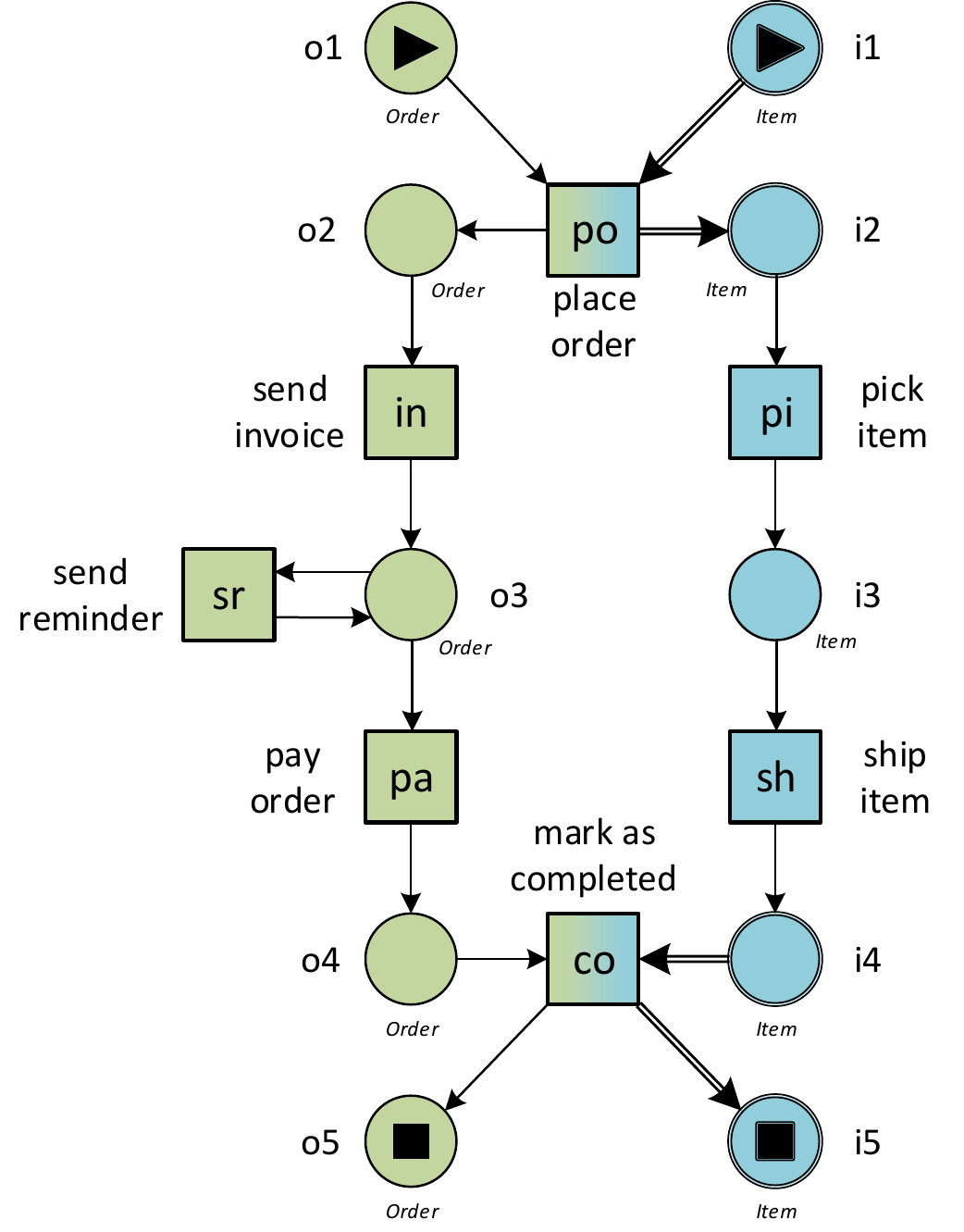}}
\caption{An object-centric Petri nets with two object types: $\mi{Order}$ and $\mi{Item}$.}
\label{f-example-pm-2}
\end{figure}

Figure~\ref{f-example-pm-2} shows the object-centric Petri net we want to discover based on the event data in Table~\ref{tablogtwoobjecttypes}.
There are now two types of places: the places that correspond to orders (colored green) and the places that correspond to items (colored blue).
Transitions are colored based on the object types they refer to. Note that transitions $\mi{po}$ and $\mi{co}$ have two colors.
A transition may consume multiple tokens from a place or produce multiple tokens for a place.
The places and arcs involved in events that consume or produce multiple objects have compound double arrows to highlight this.
Transition $\mi{po}$ in Figure~\ref{f-example-pm-2} consumes one order object from place $\mi{o1}$ and a variable number of items from place $\mi{i1}$.
$\mi{po}$ produces one order object for place $\mi{o2}$ and a variable number of items for place $\mi{i2}$.
Transition $\mi{pi}$ in Figure~\ref{f-example-pm-2} consumes one item object from place $\mi{i2}$ and produces one item object for place $\mi{i3}$.
The items are also shipped individually.
However, transition $\mi{co}$ in Figure~\ref{f-example-pm-2} consumes one order object from place $\mi{o4}$ and all items corresponding to the order from place $\mi{i4}$.
\begin{table}[htb!]
\caption{A small fragment of a simple event log with three types of objects.}\label{tablogthreeobjecttypes}
\centering
\resizebox{0.8\columnwidth}{!}{\begin{tabular}{|c|c|c|c|c|}
\hline
activity & timestamp & order & item & route\\ \hline
 $\ldots$ &  $\ldots$ &  $\ldots$ &  $\ldots$ &  $\ldots$\\
\emph{place order} & 25-11-2019:09.35 & $\{99001\}$ & $\{88124,88125,88126\}$ & $\emptyset$\\
\emph{place order} & 25-11-2019:11.35 & $\{99002\}$ & $\{88127,88128\}$ & $\emptyset$\\
 $\ldots$ &  $\ldots$ &  $\ldots$ &  $\ldots$ &  $\ldots$\\
\emph{start route} & 25-11-2019:11.35 & $\emptyset$ & $\{88124,88127\}$ & $\{66222\}$\\
\emph{end route} & 25-11-2019:11.35 & $\emptyset$ & $\{88124,88127\}$ & $\{66222\}$\\
 $\ldots$ &  $\ldots$ &  $\ldots$ &  $\ldots$ &  $\ldots$\\
\emph{start route} & 25-11-2019:11.35 & $\emptyset$ & $\{88125,88126,88128\}$ & $\{66223\}$\\
\emph{end route} & 25-11-2019:11.35 & $\emptyset$ & $\{88125,88126,88128\}$ & $\{66223\}$\\
 $\ldots$ &  $\ldots$ &  $\ldots$ &  $\ldots$ &  $\ldots$\\
\emph{mark as completed} & 01-12-2019:09.35 & $\{99001\}$ & $\{88124,88125,88126\}$& $\emptyset$\\
\emph{mark as completed} & 04-12-2019:11.05 & $\{99002\}$ & $\{88127,88128\}$& $\emptyset$\\
 $\ldots$ &  $\ldots$ &  $\ldots$ &  $\ldots$ &  $\ldots$\\
 \hline
\end{tabular}}
\end{table}

Although existing discovery techniques cannot handle the event data in Table~\ref{tablogtwoobjecttypes}, this is still a relatively simple scenario since there is a one-to-many relationship between orders and items.
In real-life applications, there may also be many-to-many relationships.
To illustrate this, consider the event log fragment depicted in Table~\ref{tablogthreeobjecttypes} where we added routes.
On any particular route, multiple items can be delivered.
The \emph{ship item} activity is now replaced by the \emph{start route} and \emph{end route} activities that may refer to items from different orders.
As shown in Table~\ref{tablogthreeobjecttypes}, route 66222 refers to two items (88124 and 88127) belonging to orders 99001 and 99002.
Route 66223 refers to three items (88125, 88126, and 88128) belonging to orders 99001 and 99002.
Again it is obvious that this cannot be modeled using traditional process models that assume a single case notion.
\begin{figure}[htb]
\centerline{\includegraphics[width=10cm]{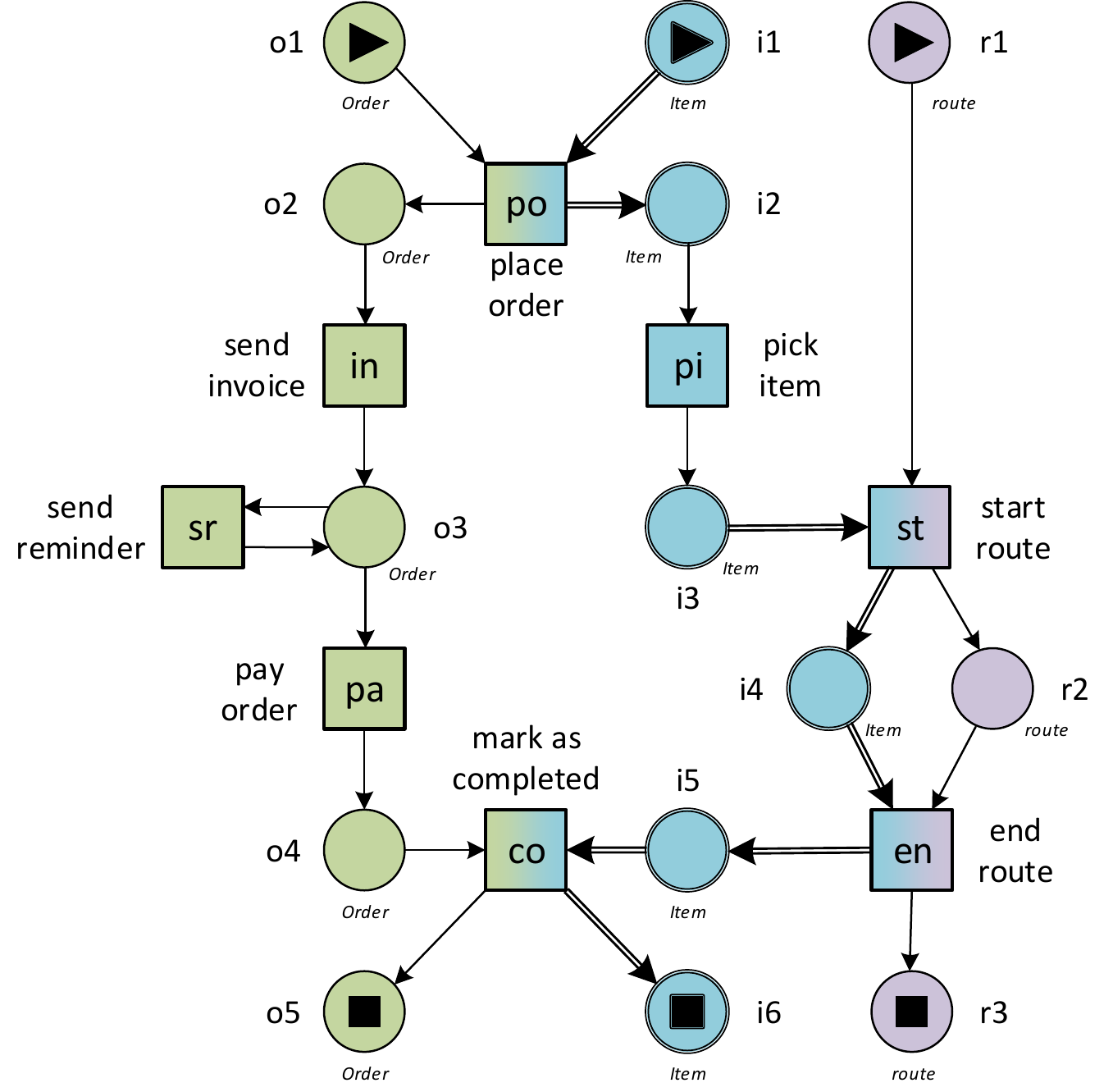}}
\caption{An object-centric Petri nets with three object types: $\mi{Order}$, $\mi{Item}$, and $\mi{Route}$.}
\label{f-example-pm-3}
\end{figure}

Figure~\ref{f-example-pm-3} shows the object-centric Petri net discovered from the event log referred to by Table~\ref{tablogthreeobjecttypes}.
There are now three types of places:  \emph{order} places (colored green), \emph{item} places (colored blue), and \emph{route} places (purple).
Transition $\mi{st}$ in Figure~\ref{f-example-pm-3} consumes a variable number of item objects from place $\mi{i3}$ and one route object from place $\mi{r1}$.
$\mi{st}$ produces a variable number of item objects for place $\mi{i4}$ and one route object for place $\mi{r2}$.
The coloring of the transitions and places and the two different types of arcs show the behaviors observed in the event log.

The problem is that existing process mining techniques assume a ``flattened event log'' where each event refers to precisely one case.
However, we would like to see process models such as the one depicted in Figure~\ref{f-example-pm-3}.
One quickly encounters the problems described in this section when applying process mining to ERP systems from SAP, Oracle, Microsoft, and other vendors of enterprise software.

\subsection{Formalizing Object-Centric Event Logs}
\label{subsec:formallog}

Tables \ref{tablogtwoobjecttypes} and \ref{tablogthreeobjecttypes} illustrate the type of data we use as input for discovery.
Such data are in-between the real data in information systems (e.g., multiple tables in a relational database) and
the traditional event data stored in the \emph{eXtensible Event Stream} (XES) format \cite{XES-standard-2013}.
Whereas XES requires one case identifier per event, our format supports any number of objects of different types per event.
To define our \emph{object-centric event logs}, we first define several universes used in the remainder (based on \cite{wvda-keynote-SEFM2019}).

\begin{definition}[Universes]\label{def:univ}
We define the following universes to be used throughout the paper:
\begin{itemize}
  \item $\univ{ei}$ is the universe of event identifiers,
  \item $\univ{act}$ is the universe of activity names (also used to label transitions in an accepting Petri net),
  \item $\univ{time}$ is the universe of timestamps,
  \item $\univ{ot}$ is the universe of object types (also called classes),
  \item $\univ{oi}$ is the universe of object identifiers (also called entities),
  \item $\mi{type} \in \univ{oi} \rightarrow \univ{ot}$ assigns precisely one type to each object identifier,
  \item $\univ{omap} = \{ \mi{omap} \in \univ{ot}  \not\rightarrow \pow(\univ{oi}) \mid \forall_{\mi{ot}\in \mi{dom}(\mi{omap})}\ \forall_{\mi{oi}\in \mi{omap}(\mi{ot})}\ \mi{type}(\mi{oi}) = \mi{ot} \}$ is the universe of all object mappings indicating which object identifiers are included per type,\footnote{$\pow(\univ{oi})$ is the powerset of the universe of object identifiers, i.e., objects types are mapped onto sets of object identifiers. $\mi{omap} \in \univ{ot} \not\rightarrow \pow(\univ{oi})$ is a partial function. If $\mi{ot} \not\in \mi{dom}(\mi{omap})$, then we assume that $\mi{omap}(\mi{ot}) = \emptyset$.}
  \item $\univ{att}$ is the universe of attribute names,
  \item $\univ{val}$ is the universe of attribute values,
  \item $\univ{vmap} = \univ{att} \not\rightarrow \univ{val}$ is the universe of value assignments,\footnote{$\univ{att} \not\rightarrow \univ{val}$ is the set of all partial functions mapping a subset of attribute names onto the corresponding values.} and
  \item $\univ{event} = \univ{ei} \times \univ{act} \times \univ{time} \times \univ{omap} \times \univ{vmap}$ is the universe of events.
\end{itemize}
\end{definition}

An event $e = (\mi{ei},\mi{act},\mi{time},\mi{omap},\mi{vmap}) \in \univ{event}$ is characterized by
a unique event identifier $ \mi{ei}$,
the corresponding activity $\mi{act}$,
the event's timestamp $\mi{time}$,
and two mappings $ \mi{omap}$ and $\mi{vmap}$ for respectively object references and attribute values.

\begin{definition}[Event Projection]\label{def:evproj}
Given $e = (\mi{ei},\mi{act},\mi{time},\mi{omap},\mi{vmap}) \in \univ{event}$,
$\pim{ei}(e)= \mi{ei}$,
$\pim{act}(e)= \mi{act}$,
$\pim{time}(e)= \mi{time}$,
$\pim{omap}(e)= \mi{omap}$, and
$\pim{vmap}(e)= \mi{vmap}$.
\end{definition}

$\pim{omap}(e) \in \univ{ot} \not\rightarrow \pow(\univ{oi})$ maps a subset of object types onto sets of object identifiers for an event $e$.
Consider for example the first visible event in Table~\ref{tablogthreeobjecttypes} and assume this is $e$.
$\pim{omap}(e)(\mi{Order}) = \{ 99001\}$, $\pim{omap}(e)(\mi{Item}) = \{ 88124,88125,88126\}$, and $\pim{omap}(e)(\mi{Route}) = \emptyset$.
Moreover, $\pim{act}(e) =$ \emph{place order} and $\pim{time}(e) =$ 25-11-2019:09.35. $\mi{dom}(\pim{vmap}(e)) = \emptyset$ since  no attribute values are mentioned in Table~\ref{tablogthreeobjecttypes}.
If the event would have a cost of 30 euros and location Aachen, then $\pim{vmap}(e)(\mi{cost}) = 30$ and $\pim{vmap}(e)(\mi{location}) =$ \emph{Aachen}.

An \emph{object-centric event log} is a collection of \emph{partially ordered events}. Event identifiers are unique, i.e., two events cannot have the same event identifier.

\begin{definition}[Object-Centric Event Log]\label{def:el}
$L=(E,\preceq_E)$ is an event log with $E \subseteq \univ{event}$ and $\preceq_E\ \subseteq E \times E$ such that:
\begin{itemize}
\item $\preceq_E$ defines a partial order (reflexive, antisymmetric, and transitive),
\item $\forall_{e_1,e_2 \in E} \ \pim{ei}(e_1)=\pim{ei}(e_2) \ \Rightarrow \ e_1 = e_2$, and
\item $\forall_{e_1,e_2 \in E} \ e_1 \preceq_E e_2 \ \Rightarrow \ \pim{time}(e_1) \leq \pim{time}(e_2)$.
\end{itemize}
\end{definition}

Definition~\ref{def:el} allows for partially ordered event logs.
However, in practice, we often use a total order, e.g., events are ordered based on timestamps and when two events have the same timestamp we assume some order.
In the tabular format used before (e.g., Table~\ref{tablogthreeobjecttypes}) we were also forced to use a total order.
However, there are process discovery techniques that take into account causalities \cite{wvda-keynote-SEFM2019,Xixi-Conf-Check-bpi2014-lnbip2015}. These can exploit such partial orders.

\section{Discovering Petri Nets for a Single Object Type}
\label{sec:discot}

Object-centric event logs generalize the traditional event log notion where each event has precisely one case identifier.
We can mimic such logs using a special object type $\mi{case} \in \univ{ot}$ such that $\card{\pim{omap}(e)(\mi{case})}=1$ for any event $e\in E$.
Since traditional process mining techniques assume this, it is common practice to convert event data with events referring to a variable number of objects to classical event logs by ``flattening'' the event data.
Assume that we take a specific object type as a case identifier. If an event has multiple objects of that type, then we can simply create one event for each object.
If an event has no objects of that type, then we simply omit the event. If an event has precisely one object of the selected type, then we keep that event. This can be formalized as follows.

\begin{definition}[Flattening Event Logs]\label{def:flat}
Let $L=(E,\preceq_E)$ be an object-centric event log and $\mi{ot} \in \univ{ot}$ an object type serving as a case notion.
The flattened event log is $L^{\mi{ot}} = (E^{\mi{ot}},\preceq_E^{\mi{ot}})$ with:\footnote{$f' = f\oplus(x,y)$ is a function such that $\mi{dom}(f') = \mi{dom}(f) \cup \{x\}$, $f'(x) = y$ and $f'(z)=f(z)$ for $z \in \mi{dom}(f)\setminus \{x\}$.}
\begin{itemize}
\item $e_i = ((\pim{ei}(e),i),\pim{act}(e),\pim{time}(e),\pim{omap}(e)\oplus(\mi{case},\{i\}),\pim{vmap}(e))$ for any $e \in E$ and $i \in \pim{omap}(e)(\mi{ot})$,
\item $E^{\mi{ot}} = \{ e_i \mid e \in E \ \wedge \ i \in \pim{omap}(e)(\mi{ot})\}$, and
\item $\preceq_E^{\mi{ot}} = \{ (e'_i,e''_j) \in E^{\mi{ot}} \times E^{\mi{ot}} \mid e' \in E \ \wedge \ i \in \pim{omap}(e')(\mi{ot})\ \wedge \  e'' \in E \ \wedge \ j \in \pim{omap}(e'')(\mi{ot}) \ \wedge \ e' \preceq_E e'' \ \wedge \ (e'= e'' \Rightarrow i=j)\}$.
\end{itemize}
\end{definition}

A flattened event log is still an event log after removing and duplicating events.

\begin{lemma}\label{lem:flat}
Let $L=(E,\preceq_E)$ be an object-centric event log and $\mi{ot} \in \univ{ot}$ an object type serving as a case notion.
The flattened event log $L^{\mi{ot}} = (E^{\mi{ot}},\preceq_E^{\mi{ot}})$ is indeed an event log as defined in Definition~\ref{def:el}.
\end{lemma}
\begin{proof}
$\preceq_E^{\mi{ot}}$ defines a partial order.
For any $e_i\in E^{\mi{ot}}$, $e_i \preceq_E^{\mi{ot}} e_i$ (reflexive).
If $e'_i \preceq_E^{\mi{ot}} e''_j$ and $e''_j \preceq_E^{\mi{ot}} e'_i$, then $e'= e''$ and $i=j$, and hence also $e'_i = e''_j$ (antisymmetric).
If $e'_i \preceq_E^{\mi{ot}} e''_j$ and $e''_j \preceq_E^{\mi{ot}} e'''_k$, then $e' \preceq_E e''$, $e'' \preceq_E e'''$, $(e'= e'' \Rightarrow i=j)$, and $(e''= e''' \Rightarrow j=k)$.
Hence, $e' \preceq_E e'''$ ($\preceq_E$ is transitive).
If $e' \neq e'''$, then $e'_i \preceq_E^{\mi{ot}} e'''_k$ due to the definition of $\preceq_E^{\mi{ot}}$.
If $e' = e'''$, then $e'= e''$ and $e''=e'''$. Hence, $i=j$ and $j=k$ (see above), and $i=k$.
Again we conclude that $e'_i \preceq_E^{\mi{ot}} e'''_k$ (transitive).
If $\pim{ei}(e'_i)=\pim{ei}(e''_j)$, then $(e',i) = (e'',j)$ making event identifiers unique.
If $e'_i \preceq_E^{\mi{ot}} e''_j$, then $e' \preceq_E e''$. Hence, $\pim{time}(e'_i) = \pim{time}(e')  \leq \pim{time}(e'') = \pim{time}(e''_j)$ showing that time cannot go backwards.
\end{proof}

Table~\ref{tablogthreeobjecttypes} shows eight events (the rest is omitted). Assume $L = (E,\preceq_E)$ is the log consisting of only these eight events.
The flattened event log  $L^{\mi{Order}} = (E^{\mi{Order}},\preceq_E^{\mi{Order}})$ has four events (the four middle events in Table~\ref{tablogthreeobjecttypes} are removed).
The flattened event log $L^{\mi{Item}} = (E^{\mi{Item}},\preceq_E^{\mi{Item}})$ has 20 events since all original events are replicated two or three times.
The flattened event log  $L^{\mi{Route}} = (E^{\mi{Route}},\preceq_E^{\mi{Route}})$ has four events.

Assume now that $L=(E,\preceq_E)$ is flattened using object type $\mi{ot}$ leading to event log $L^{\mi{ot}} = (E^{\mi{ot}},\preceq_E^{\mi{ot}})$.
We then have a conventional event log with a selected case notion and can apply all existing process mining techniques.
However, flattening the event log using $\mi{ot}$ as a case notion potentially leads to the following problems.
\begin{itemize}
  \item \emph{Deficiency}: Events in the original event log that have \emph{no} corresponding events in the flattened event log disappear from the data set (i.e., $\pim{omap}(e)(\mi{ot}) = \emptyset$).
  \item \emph{Convergence}: Events referring to \emph{multiple} objects of the selected type are replicated, possibly leading to unintentional duplication (i.e., $\card{\pim{omap}(e)(\mi{ot})} \geq 2$).
  \item \emph{Divergence}: Events referring to \emph{different} objects of a type \emph{not} selected as the case notion are considered to be causally related. For example, two events refer to the same order but different times or two events refer to the same route but different items.
\end{itemize}

\begin{definition}[Deficiency, Convergence, and Divergence]\label{def:fltproblems}
Let $L=(E,\preceq_E)$ be an object-centric event log and $L^{\mi{ot}} = (E^{\mi{ot}},\preceq_E^{\mi{ot}})$ the
flattened event log based on object type $\mi{ot} \in \univ{ot}$.
Event $e \in E$ has a deficiency problem if $\pim{omap}(e)(\mi{ot}) = \emptyset$ (i.e., the event is ignored when using $\mi{ot}$ as case notion).
Event $e \in E$ has a convergence problem if $\card{\pim{omap}(e)(\mi{ot})} \geq 2$ (i.e., the event is unintentionally replicated when using $\mi{ot}$ as case notion).
Event $e \in E$ has a divergence problem if there exist another event 
$e' \in E$ and object type $\mi{ot}' \in \univ{ot}$
such that $\pim{omap}(e)(\mi{ot}) \neq \emptyset$, $\pim{omap}(e')(\mi{ot}) \neq \emptyset$,
$\pim{omap}(e)(\mi{ot}') \neq \emptyset$, $\pim{omap}(e')(\mi{ot}') \neq \emptyset$,
$\pim{omap}(e)(\mi{ot}) = \pim{omap}(e')(\mi{ot})$, and $\pim{omap}(e)(\mi{ot}') \neq \pim{omap}(e')(\mi{ot}')$.
\end{definition}

Note that in case of divergence, there are two events $e$ and $e'$ and two candidate case notions $\mi{ot}$
and $\mi{ot}'$ such that both events refer to objects of both object types 
and the events ``agree'' on $\mi{ot}$ but not on $\mi{ot}'$.

Consider again the eight events shown in Table~\ref{tablogthreeobjecttypes}. When taking \emph{Order} or \emph{Route} as the object type used to flatten the event log, half the events disappear (deficiency).
When taking \emph{Item} as the object type used to flatten the event log, the first event is replaced by three \emph{place order} events, the second event is replaced by two \emph{place order} events, etc. This is misleading since these replicated events occurred only once (convergence).
To explain divergence, assume that an order consists of 10 items and object type \emph{Order} is used to flatten the event log.
There will be 10 pick events that are executed in a given order. Although they are independent, they will seem to be causally related (same case) and most discovery algorithms will introduce a loop, although there is precisely one pick event per item.
\begin{figure}[htb]
\centerline{\includegraphics[width=12cm]{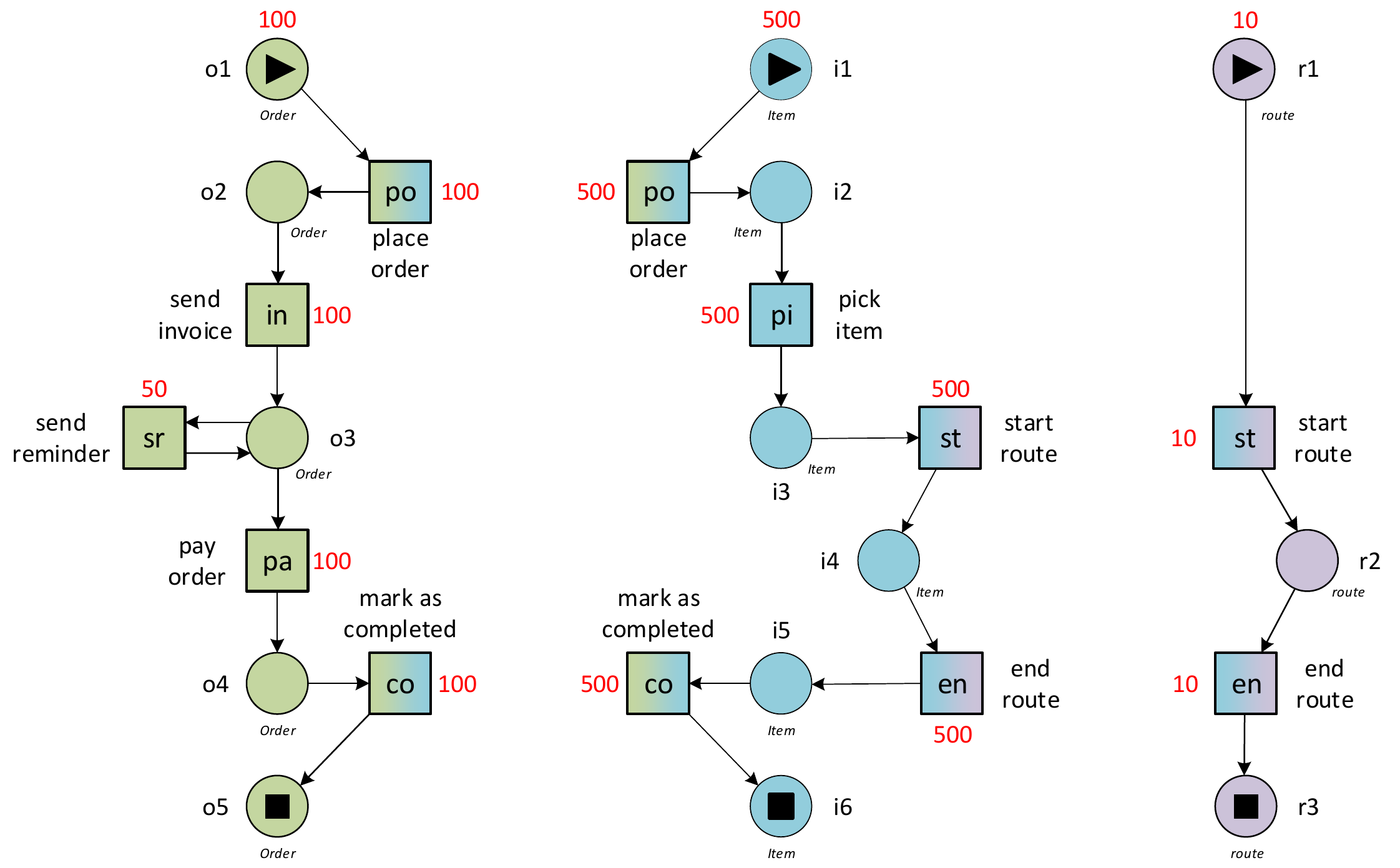}}
\caption{Three accepting Petri nets discovered for the three flattened event logs: $(E^{\mi{Order}},\preceq_E^{\mi{Order}})$ (left), $(E^{\mi{Item}},\preceq_E^{\mi{Item}})$ (middle), and $(E^{\mi{Route}},\preceq_E^{\mi{Route}})$ (right). The numbers in red refer to the total number of tokens produced or consumed per arc.}
\label{f-example-pm-4}
\end{figure}

Figure~\ref{f-example-pm-4} shows three process models discovered for three flattened event logs: $L^{\mi{Order}} = (E^{\mi{Order}},\allowbreak \preceq_E^{\mi{Order}})$, $L^{\mi{Item}} = (E^{\mi{Item}},\allowbreak \preceq_E^{\mi{Item}})$, and $L^{\mi{Route}} = (E^{\mi{Route}},\allowbreak \preceq_E^{\mi{Route}})$.
For example, the accepting Petri net in the middle was discovered based on $L^{\mi{Item}}$, i.e., the original event log flattened using object type \emph{Item}.
Assume that there are 100 orders with on average 5 items per order. This implies that there are 500 items. Assume that each route consists, on average, of 50 items that need to be delivered, i.e., there are 10 routes in total.
These numbers are depicted in Figure~\ref{f-example-pm-4}. Although the three accepting Petri nets look reasonable, they do not ``fit'' together (the frequencies of the corresponding activities are different).
For example, in the left model (order) the \emph{place order} activity is performed 100 times and in the middle model (item) the same activity is executed 500 times (factor 5).
In the right model (route) the \emph{start route} activity is performed 10 times and in the middle model (item) the same activity is executed 500 times (factor 50).
These mismatches illustrate the convergence problem. One could argue that the accepting Petri net in the middle is wrong because the frequencies of activities do not match the frequencies in the original process model.
\begin{figure}[htb]
\centerline{\includegraphics[width=15cm]{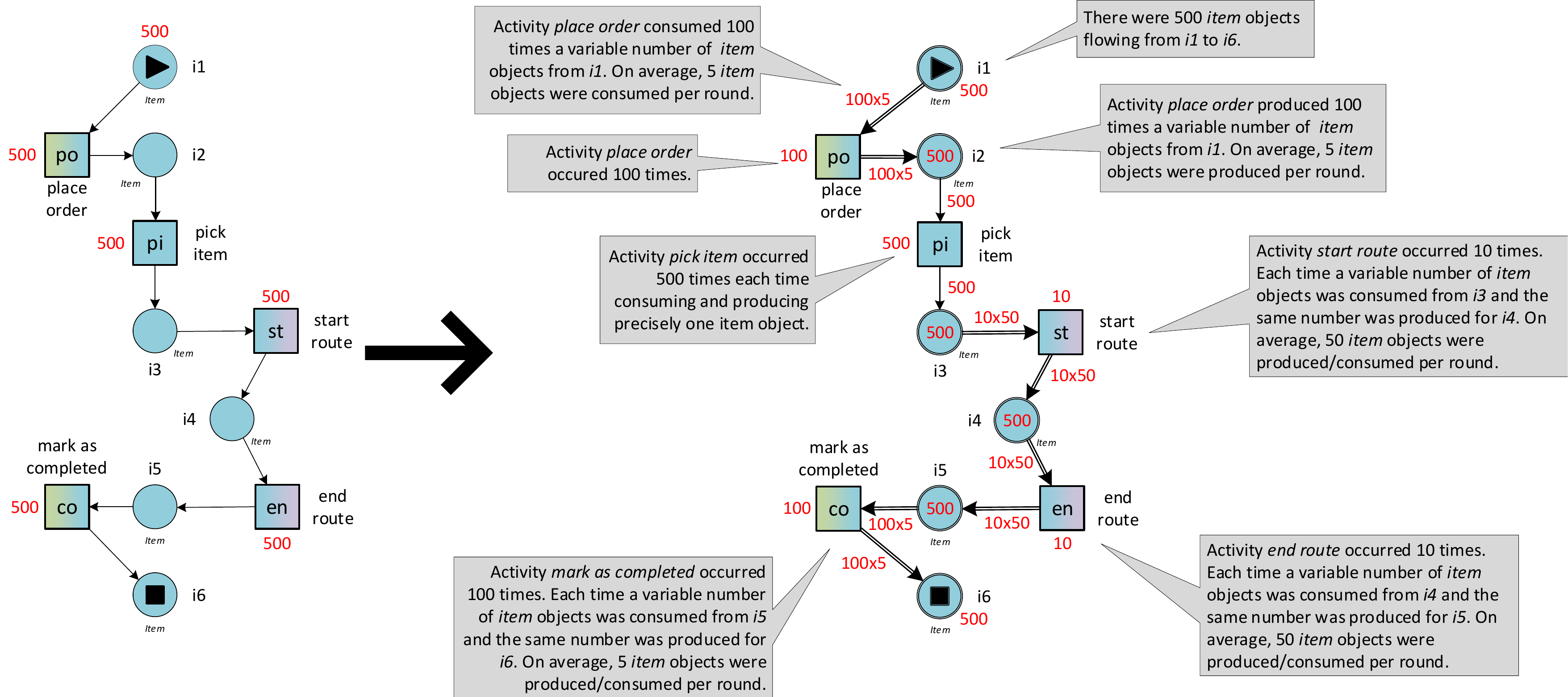}}
\caption{The model on the left was discovered for $L^{\mi{Item}} = (E^{\mi{Item}},\preceq_E^{\mi{Item}})$. Because of flattening, the frequencies of activities are not correct. However, it is known which transition occurrences belonged to each event and we can regroup them. This can be used to merge occurrences, leading to the process model on the right.}
\label{f-example-pm-5}
\end{figure}

Figure~\ref{f-example-pm-5} sketches how the problem of incorrect activity frequencies can be resolved using \emph{variable arcs}, i.e., arcs that can be used to consume or produce multiple tokens in one step.
Such ``multiset arcs'' are also possible in colored Petri nets \cite{mbp-aal-stahl-2011,cpnbook-jensen-2009}.
When an event was replicated to produce the ``flat model'' (e.g., Figure~\ref{f-example-pm-4}), we can merge the corresponding transition occurrences into one transition occurrence that may consume and produce multiple tokens.
See for example transition \emph{place order}. In the accepting Petri net on the left, transition \emph{place order} fires 500 times when replaying the flattened event log $L^{\mi{Item}}$.
However, we know exactly which transition occurrences belong together. This can be used to reconstruct transition occurrences that consume and produce a variable number of tokens in one step.
For transition \emph{place order} this means that 500 occurrences are merged onto 100 occurrences that, on average, consume and produce 5 tokens per arc. To indicate this, we use compound double arrows with the annotation $100 \times 5$.
Next, consider transition \emph{start route}. In the model on the left, transition \emph{start route} fires 500 times.
However, we know exactly which of these 500 transition occurrences belong to the 10 routes. Again these low-level transition occurrences can be merged into higher-level transition occurrences that consume and produce a variable number of tokens in one step.
For transition \emph{start route} this means that there are 10 occurrences that, on average, consumer and produce 50 tokens per arc. To indicate this, we use again compound double arrows, but now with the annotation $10 \times 50$.
Only the occurrences of \emph{pick item} did not change due to flattening. Hence, the corresponding arcs did not change.

Figure~\ref{f-example-pm-5} sketches how we can create Petri nets for one object type where the frequency of each transition matches the actual number of corresponding events in the event log.
These models can be merged into more holistic process models showing the different object types as is shown next.

\section{Object-Centric Petri Nets}
\label{sec:ocpn}

As indicated in the previous sections, we need to be able to distinguish the \emph{different object types} and a single event (i.e., transition occurrence)
may involve a \emph{variable number of objects} (e.g., one order may have any number of items).
An obvious way to model such processes is to use colored Petri nets where places can have different types \cite{mbp-aal-stahl-2011,cpnbook-jensen-2009}.
Figure~\ref{f-cpn} shows a screenshot of CPN Tools while simulating the scenario with 100 orders, 500 items, and 10 routes described before.
The color sets \emph{Order}, \emph{Item}, and \emph{Route} are used to type the places.
The ten arcs with the annotation $or$ produce or consume a single order.
The two arcs with the annotation $it$ produce or consume a single item.
The four arcs with the annotation $rt$ produce or consume a single route.
There are eight arcs with the annotation $its$ which is a variable of type \emph{Items}, i.e., a list of items.
These consume or produce a variable number of item objects.
The four guards determine the correspondence between orders, routes, and items.
For example, the guard $[its=oi(or)]$ of transition \emph{place order} specifies the set of items $its$ involved in a specific order $or$. The same guard is used for transition \emph{mark as completed}.
Transitions \emph{start route} and \emph{end route} use guard $[its=ri(rt)]$ to determine the items $its$ involved in route $rt$.
\begin{figure}[htb]
\centerline{\includegraphics[width=14cm]{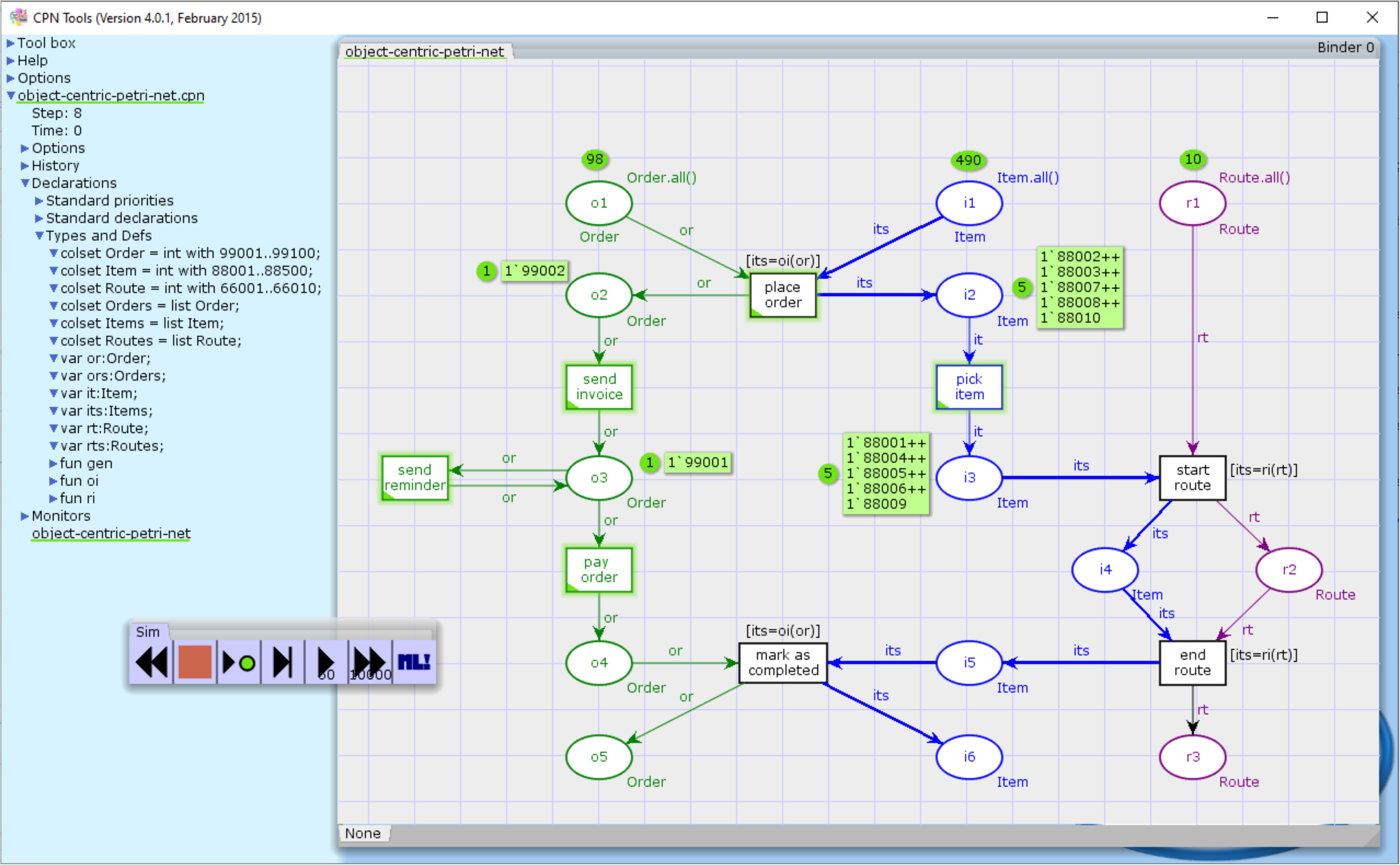}}
\caption{A colored Petri net in CPN Tools \cite{mbp-aal-stahl-2011,cpnbook-jensen-2009} modeling the process depicted in Figure~\ref{f-example-pm-3} which was discovered from the event data in Table~\ref{tablogthreeobjecttypes}.}
\label{f-cpn}
\end{figure}

Figure~\ref{f-cpn} shows that one can model processes involving multiple objects using colored Petri nets (or related formalisms).
However, it is infeasible to discover an \emph{arbitrary} colored Petri net from an (object-centric) event log.
We need a \emph{representational bias} that corresponds to the information in the event log.
Therefore, we aim to discover a specific type of colored Petri net.
To simplify matters, we also abstract from the matching between the different objects (i.e., the guards in Figure~\ref{f-cpn}).
This allows us to use a more specific and more abstract representation called \emph{object-centric Petri net}.\footnote{Terms similar to ``object Petri nets'' were already used by R\"{u}diger Valk, Charles Lakos, Jinzhong Niu, Li-Chi Wang, Daniel Moldt, and others. Note that our nets are different and some overloading of terminology is unavoidable.}

\begin{definition}[Object-Centric Petri Net]\label{def:oopn}
An \emph{object-centric Petri net} is a tuple $\mi{ON}=(N,\mi{pt},F_{\mi{var}})$ where
$N=(P,T,F,l)$ is a labeled Petri net, $\mi{pt} \in P \rightarrow \univ{ot}$ maps places onto object types, and $F_{\mi{var}}\subseteq F$ is the subset of variable arcs.
\end{definition}

Figure~\ref{f-example-pm-3} shows an object-centric Petri net:
$P = \{o1,\ldots,o5,i1,\ldots,i6,r1,r2,r3\}$, $T = \{\mi{po},\allowbreak \mi{in},\allowbreak \mi{pi}, \ldots\}$,
$F = \{(o1,\mi{po}),(i1,\mi{po}),(\mi{po},o2),(\mi{po},i2), \ldots\}$,
$l(\mi{po})=$ \emph{place order}, $l(\mi{in})=$ \emph{send invoice}, etc.,
$\mi{pt}(o1)= \mi{Order}$, $\mi{pt}(i1)= \mi{Item}$, $\mi{pt}(r1)= \mi{Route}$, etc., and
$F_{\mi{var}} = \{(i1,\mi{po}),(\mi{po},i2), \ldots\}$.
Note that the graphical notation in Figure~\ref{f-example-pm-3} fully defines the object-centric Petri net.

\begin{definition}[Well-Formed]\label{def:wfoopn}
Let $\mi{ON}=(N,\mi{pt},F_{\mi{var}})$ be an object-centric Petri net with $N=(P,T,F,l)$. We introduce the following notations:
\begin{itemize}
  \item $\mi{pl}(t) = \pre t \cup \post t$ are the input and output places of $t\in T$, $\mi{pl}_{\mi{var}}(t) = \{p \in P \mid \{(p,t),(t,p)\} \cap F_{\mi{var}} \neq \emptyset\}$ are the input and output places connected through variable arcs, and $\mi{pl}_{\mi{nv}}(t) = \{p \in P \mid \{(p,t),(t,p)\} \cap (F \setminus F_{\mi{var}}) \neq \emptyset\}$ are the places connected through non-variable arcs.
  \item $\mi{tpl}(t) = \{pt(p) \mid p \in \mi{pl}(t)\}$, $\mi{tpl}_{\mi{var}}(t) = \{pt(p) \mid p \in \mi{pl}_{\mi{var}}(t)\}$, and $\mi{tpl}_{\mi{nv}}(t) = \{pt(p) \mid p \in \mi{pl}_{\mi{nv}}(t)\}$ are the corresponding place types.
\end{itemize}
$\mi{ON}$ is \emph{well-formed} if for each transition $t \in T$: $\mi{tpl}_{\mi{var}}(t) \cap \mi{tpl}_{\mi{nv}}(t) = \emptyset$.
\end{definition}

In a \emph{well-formed} object-centric Petri net, the arcs should ``agree'' on variability, i.e., a combination of an object type and transition has variable arcs or normal arcs but not both.
For example, because $(i1,\mi{po}) \in F_{\mi{var}}$ also $(\mi{po},i2) \in F_{\mi{var}}$.
Because $(o1,\mi{po}) \not\in F_{\mi{var}}$ also $(\mi{po},o2) \not\in F_{\mi{var}}$.
This assumption is reasonable when looking at an object-centric event log. 
Per event $e$ and object type $\mi{ot}$, $\pim{omap}(e)(\mi{ot})$ is given.
Therefore, it makes no sense to consider different sets of objects of the same type $\mi{ot}$ per transition $t$.
In the remainder, we limit ourselves to well-formed object-centric Petri nets (without explicitly stating this).

A token denoted by $(p,\mi{oi})$ resides in place $p$ and refers to object $\mi{oi}$. A marking is a multiset of such tokens.
In the marking $[(p1,666),(p2,666),(p2,555),(p3,555)]$ there are four tokens (place $p2$ has two tokens referring to objects 555 and 666).
\begin{definition}[Marking]\label{def:oomark}
Let $\mi{ON}=(N,\mi{pt},F_{\mi{var}})$ be an object-centric Petri net with $N=(P,T,F,l)$.
$Q_{\mi{ON}} = \{(p,\mi{oi}) \in P \times \univ{oi} \mid \mi{type}(\mi{oi}) = \mi{pt}(p)\}$ is the set of possible tokens.
A marking $M$ of $\mi{ON}$ is a multiset of tokens, i.e., $M \in \bag(Q_{\mi{ON}})$.
\end{definition}

To describe the semantics of an object-centric Petri net, we use the notion of \emph{bindings}, similar to the notion of bindings in colored Petri nets.
However, now the binding refers to the object references of the corresponding event in the event log.
A binding $(t,b)$ refers to a transition $t$ and a function $b$ that maps a subset of object types to sets of object identifiers.
The subset of object types corresponds to the object types of the surrounding places (i.e., $\mi{tpl}(t)$).
Moreover, for non-variable arcs the binding should select precisely one object (i.e., $\card{b(\mi{ot})}=1$ for $\mi{ot} \in \mi{tpl}_{\mi{nv}}(t)$).
Consider transition $t$ and one of its input place $p$ (i.e., $p \in \pre t$).
If $t$ fires with binding $(t,b)$, then $\mi{pt}(p) \in \mi{dom}(b)$ and the objects $b(\mi{pt}(p))$ are removed from input place $p$.
If $p$ is an output place of $t$ ($p \in \post t$), then the objects $b(\mi{pt}(p))$ are added to output place $p$.
Therefore, binding $(t,b)$ fully determines the new marking.

\begin{definition}[Binding Execution]\label{def:oocctrans}
Let $\mi{ON}=(N,\mi{pt},F_{\mi{var}})$ be an object-centric Petri net with $N=(P,T,F,l)$.
$B = \{ (t,b) \in T \times \univ{omap} \mid \mi{dom}(b) = \mi{tpl}(t) \ \wedge \ \allowbreak \forall_{\mi{ot} \in \mi{tpl}_{\mi{nv}}(t)}\ \allowbreak \card{b(\mi{ot})}=1 \}$  is the set of all possible bindings.
$(t,b) \in B$ is a binding and corresponds to the execution of transition $t$ consuming selected objects from the input places and producing the corresponding objects for the output places (both specified by $b$).
$\mi{cons}(t,b) = [ (p,\mi{oi}) \in Q_{\mi{ON}} \mid p \in \pre t \ \wedge \ \mi{oi} \in b(\mi{pt}(p))]$ 
is the multiset of tokens to be consumed given binding $(t,b)$.
$\mi{prod}(t,b) = [ (p,\mi{oi}) \in Q_{\mi{ON}} \mid p \in \post t \ \wedge \ \mi{oi} \in b(\mi{pt}(p))]$
is the multiset of tokens to be produced given binding $(t,b)$.
Binding $(t,b)$ is \emph{enabled} in marking $M \in \bag(Q_{\mi{ON}})$ if $\mi{cons}(t,b) \leq M$.
The occurrence of an enabled binding $(t,b)$ in marking $M$ leads to the new marking $M' = M - \mi{cons}(t,b) + \mi{prod}(t,b)$.\footnote{Summation ($+$), difference ($-$), and inclusion ($\leq$) are defined for multisets in the usual way, e.g., $[a,b] + [b,c] = [a,b^2,c]$, $[a,b^2,c] - [b,c] = [a,b]$, and $[a,b] \leq [a,b^2,c]$.} This is denoted as $M \stackrel{(t,b)}{\longrightarrow} M'$.
\end{definition}

$M \stackrel{(t,b)}{\longrightarrow} M'$ implies that binding $(t,b)$ is enabled in marking $M$ and that the occurrence of this binding leads to the new marking $M'$.
It is also possible to have a sequence of enabled bindings $\sigma = \langle (t_1,b_1),(t_2,b_2), \ldots ,(t_n,b_n)\rangle \in B^*$ such that
$M_0 \stackrel{(t_1,b_1)}{\longrightarrow} M_1 \stackrel{(t_2,b_2)}{\longrightarrow} M_2 \stackrel{(t_3,b_3)}{\longrightarrow} \ldots \stackrel{(t_n,b_n)}{\longrightarrow} M_n$, i.e., it is possible to reach $M_n$ from $M_0$ in $n$ steps. This is denoted $M \stackrel{\sigma}{\longrightarrow} M'$.
It is also possible to map the transition names onto the corresponding activity names using $l$ leading to the so-called \emph{visible binding sequence} $\sigma_v = \langle (l(t_1),b_1),(l(t_2),b_2), \ldots ,(l(t_n),b_n)\rangle$ (where $(l(t_i),b_i)$ is omitted if $t_i$ has no label).
Note that the visible binding sequence does not show silent steps (transitions with no label) and cannot distinguish duplicate activities (two transitions with the same label).

It should be noted that Definition~\ref{def:oocctrans} does not put any constraints on the binding other than that for non-variable arcs precisely one token is consumed/produced.
In the colored Petri in Figure~\ref{f-cpn} there are four transitions with guards to link items to specific orders and routes.
This is deliberately abstracted from in Definition~\ref{def:oopn} to enable the discovery of so-called \emph{accepting object-centric Petri nets} with an initial and final marking.

\begin{definition}[Accepting Object-Centric Petri Net]\label{def:aoopn}
An accepting object-centric Petri net is a tuple $\mi{AN}=(\mi{ON},M_{\mi{init}},M_{\mi{final}})$ composed of a well-formed object-centric Petri net $\mi{ON}=(N,\mi{pt},F_{\mi{var}})$,
an initial marking $M_{\mi{init}} \in \bag(Q_{\mi{ON}})$, and a final marking $M_{\mi{final}} \in \bag(Q_{\mi{ON}})$.
\end{definition}

Using the notion of a visible binding sequence, we can reason about all behaviors leading from the initial to the final marking.

\begin{definition}[Language of an Object-Centric Petri Net]\label{def:laocpn}
An accepting object-centric Petri net $\mi{AN}=(\mi{ON},M_{\mi{init}},M_{\mi{final}})$ defines a language
$\phi(\mi{AN}) = \{ \sigma_v \mid M_{\mi{init}} \stackrel{\sigma}{\longrightarrow} M_{\mi{final}} \}$ that is composed of all visible binding sequences starting in $M_{\mi{init}}$ and ending in $M_{\mi{final}}$.
\end{definition}

Note that the behavior of an accepting object-centric Petri net is deliberately ``underspecified''. There are only typing and cardinality constraints. Hence, objects of different types are unrelated. Compared to the colored Petri net in Figure~\ref{f-cpn}, our process models do not use guards to relate objects of different types.
Note that guards combine objects of different types that are only characterized by an identifier. Using just the identifiers would lead to overfitting models. How to find a rule telling that order 99001 is composed of items 88124, 88125, and 88126?
This is contained in the data and cannot be handled by a precise and explicit rule.
As mentioned in the conclusion, this a topic for future research (cf.\ Section~\ref{sec:concl}).

\section{Discovering Object-Centric Petri Nets}
\label{sec:discoopn}

First, we introduce a general approach to learn accepting object-centric Petri nets from object-centric event logs.
Then we discuss performance-related annotations of the models, model views, and ways to combine these results with traditional process mining techniques.

\subsection{Generic Approach}
\label{sec:appdiscoopn}

Given an object-centric event log $L = (E,\preceq_E)$ (Definition~\ref{def:el}), we would like to discover an accepting object-centric Petri net  $\mi{AN}=(\mi{ON},M_{\mi{init}},M_{\mi{final}})$ (Definition~\ref{def:aoopn}).
Rather than defining one specific discovery algorithm, we present a general approach leveraging existing process discovery techniques.
\begin{itemize}
  \item \textbf{Step 1:} Given an object-centric event log $L = (E,\preceq_E)$, identify the object types $\mi{OT} \subseteq \univ{ot}$ appearing in the event log. Then create a flattened event log $L^{\mi{ot}} = (E^{\mi{ot}},\preceq_E^{\mi{ot}})$ for each object type $\mi{ot} \in \mi{OT}$.

  \item \textbf{Step 2:} Discover an accepting Petri net $\mi{SN}^{\mi{ot}}=(N^{\mi{ot}},M_{\mi{init}}^{\mi{ot}},M_{\mi{final}}^{\mi{ot}})$ with $N^{\mi{ot}}=(P^{\mi{ot}},T^{\mi{ot}},\allowbreak F^{\mi{ot}},\allowbreak l^{\mi{ot}})$ for each object type $\mi{ot} \in \mi{OT}$ using the flattened event log $L^{\mi{ot}}$. For this purpose, any conventional discovery technique can be used. The only assumption we need to make is that there are no duplicated labels, i.e., labeling function $l^{\mi{ot}}$ is injective. However, we allow for silent transitions, i.e., $l^{\mi{ot}}$ may be partial.

  \item \textbf{Step 3:} Merge the accepting Petri nets into a Petri net $N$. To avoid name clashes, first ensure that the place names and names of silent transitions in the different nets are different. Also, ensure that transitions that have the same label also have the same name (this is possible because the labeling functions are injective). After renaming, create an overall labeled Petri net $N=(P,T,F,l)$ with:
      $P = \bigcup_{\mi{ot} \in \mi{OT}}  P^{\mi{ot}}$,
      $T = \bigcup_{\mi{ot} \in \mi{OT}}  T^{\mi{ot}}$,
      $F = \bigcup_{\mi{ot} \in \mi{OT}}  F^{\mi{ot}}$, and
      $l = \bigcup_{\mi{ot} \in \mi{OT}}  l^{\mi{ot}}$.

  \item \textbf{Step 4:} Assign object types to the places in the merged Petri net $N$:
      $\mi{pt}(p)=\mi{ot}$ for $p \in P^{\mi{ot}}$ and $\mi{ot} \in \mi{OT}$. This is possible because the places for the different object types are disjoint.

  \item \textbf{Step 5:} Identify the variable arcs $F_{\mi{var}}\subseteq F$. This can be determined in different ways (e.g., using replay results or diagnosing the flattening process). The goal is to identify the arcs where multiple tokens need to be consumed or produced. An example would be $F_{\mi{var}} = \{(p,t)\in F \cap (P \times T) \mid \mi{score}(l(t),\mi{pt}(p)) < \tau \} \cup \{(t,p)\in F \cap (T \times P) \mid \mi{score}(l(t),\mi{pt}(p)) < \tau \}$ where $\tau$ is a threshold (e.g., 0.98) and $\mi{score} \in (\univ{act} \times \univ{ot}) \not\rightarrow [0,1]$ such that $\mi{score}(\mi{act},\mi{ot}) = \card{\{e \in E \mid \pim{act}(e)= \mi{act} \ \wedge \ \card{\pim{omap}(e)(\mi{ot})}=1 \}} / \card{\{e \in E \mid \pim{act}(e)= \mi{act}\}}$ is the fraction of $\mi{act}$ events that refer to precisely one object of type $\mi{ot}$.

  \item \textbf{Step 6:} Combining the previous three steps allows us to create an object-centric Petri net $\mi{ON}=(N,\mi{pt},F_{\mi{var}})$. The initial and final markings are obtained by replicating the markings of the accepting Petri nets for each of the corresponding objects.
  $M_{\mi{init}} = [(p,\mi{oi}) \in Q_{\mi{ON}} \mid \exists_{\mi{ot} \in \mi{OT}}\ p \in M_{\mi{init}}^{\mi{ot}} \ \wedge \ \exists_{e\in E}\ \mi{oi} \in \pim{omap}(e)(\mi{pt}(p))]$.
  $M_{\mi{final}} = [(p,\mi{oi}) \in Q_{\mi{ON}} \mid \exists_{\mi{ot} \in \mi{OT}}\ p \in M_{\mi{final}}^{\mi{ot}} \ \wedge \ \exists_{e\in E}\ \mi{oi} \in \pim{omap}(e)(\mi{pt}(p))]$.

  \item \textbf{Step 7:} Return the accepting object-centric Petri net $\mi{AN}=(\mi{ON},M_{\mi{init}},M_{\mi{final}})$.

\end{itemize}

The above approach has two parameters: (1) the discovery technique used in Step 2 and (2) the selection of variable arcs in Step 5 (e.g., threshold $\tau$ and function $\mi{score}$).
For Step 2 any discovery technique that produces a Petri net without duplicate labels can be used (e.g., region-based techniques without label splitting or the inductive mining techniques).
The scoring function described in Step 5 is just an example. Function $\mi{score}(\mi{act},\mi{ot})$ counts the fraction of $\mi{act}$ events that refer to precisely one object of type $\mi{ot}$.
If this is rather low (below the threshold $\tau$), then the corresponding arcs are considered to be variable (i.e., these arcs can consume/produce any number of tokens).
The approach always returns a well-formed object-centric Petri net because the selection of $F_{\mi{var}}$ depends on the transition and place type only.

\subsection{Annotations, Views, and Extractions}
\label{sec:viewdiscoopn}

The main novelty of the work presented in this paper is that we discover a single process model with multiple object types allowing us to capture multiple one-to-many and many-to-many relationships in event data.
Based on this, many ideas from traditional process mining can be converted to this more realistic setting. In this section, we briefly discuss a few.

It is rather straightforward to annotate process models with frequency information and time information. For example, the right-hand side of Figure~\ref{f-example-pm-5} is already showing various frequencies
and our implementation provides much more diagnostics.
\begin{itemize}
  \item \textbf{Transition annotations:} The frequency of a transition shows how often the corresponding activity occurred in the object-centric event log. It is also possible to add statistics about the objects involved in the corresponding events (e.g., how many objects of a particular type were involved on average ).
  If there is transactional information (start and complete), it is also possible to show information about the duration of the corresponding activity  (average, median, variance, minimum, maximum, etc.).
  \item \textbf{Place annotations:} It is possible to show how many tokens have been consumed from and produced for each place. These tokens correspond to objects. Hence, it is also possible to show how many unique objects visited the place and what the average number of visits per object is. By taking the time difference between the moment a token is produced and consumed, it is possible to show timing information (average, median, variance, minimum, maximum, etc.). In case of compliance checking, one can also show missing and remaining tokens (see implementation).
  \item \textbf{Arc annotations:} There are two types of arcs: the variable arcs $F_{\mi{var}}$ and the non-variable $F \setminus F_{\mi{var}}$. Both can be annotated with frequency and time information. For variable arcs, we can also show statistics about the numbers of tokens produced/consumed per transition occurrence. See Figure~\ref{f-example-pm-5}, where the annotations for variable arcs show averages. For example, annotation $100 \times 5$ shows that 100 times a multiset of tokens was moved along the arc and the average size of this multiset was 5, indicating that 500 objects were moved along the arc.
\end{itemize}

Next to adding annotations, it is also possible to select or deselect object types. The approach described in Section~\ref{sec:appdiscoopn}
first identifies the object types $\mi{OT} \subseteq \univ{ot}$ appearing in the event log. However, we can take any nonempty subset $\mi{OT}' \subseteq \mi{OT}$.
It is, for example, possible to leave out the object type \emph{Order} and only use the types \emph{Item} and \emph{Route}. This way it is possible to create simplified \emph{views}.
Everything can also be combined with frequency-based filtering, i.e., adding sliders to seamlessly remove infrequent activities and arcs.

Since most process mining techniques cannot handle object-centric event logs, it is valuable to be able to generate classical event logs and apply traditional techniques.
The holistic view provided by the accepting object-centric Petri net serves as a starting point for a more detailed analysis focusing on one object type.
Definition~\ref{def:flat} already showed that it is easy to flatten event logs. It is also possible to take as case identifier combinations of object types.
This can be combined with views and interactive filtering. Of course, one should always be very careful when interpreting such results.
Due to the convergence and divergence problems mentioned before the results may be misleading.
However, the overall accepting object-centric Petri net helps to avoid misinterpretations.

\section{Tool Support for Object-Centric Petri Nets}
\label{sec:impl}

The concepts and techniques discussed have been fully implemented.
In this section, we describe the implementation, the functionalities supported, and evaluate the performance.

\subsection{Implementation}
\label{subsec:impl}

To support the discovery approach presented in this paper (including performance and conformance analysis using token-based replay), we extended \emph{PM4Py}
with an additional Python library \emph{PM4Py-MDL}.\footnote{The software can be downloaded from \url{www.pm4py.org} and \url{https://github.com/Javert899/pm4py-mdl.git}}
The tool can be installed by using the Python Package Installer (PIP) 
(use the command \emph{pip install pm4pymdl}).
Next to discovering object-centric Petri nets, \emph{PM4Py-MDL} can also discover multi-dimensional directly-follows graphs \cite{Alessandro-StarStar-SIMPDA-CEUR2018,Alessandro-MVP-SIMPDA-ext2019}.

Our implementation follows the approach described in this paper.
The discovery of an object-centric Petri net is based on the 
discovery of Petri nets for the single object types.
Then, these Petri nets are merged and annotated.
For the discovery of a Petri net for each of the individual object types,
a sound workflow net is obtained by applying 
the Inductive Miner Directly-Follows process discovery algorithm \cite{sander-scalable-procmin-SOSYM}. However, any discovery technique producing an accepting Petri net can be used.

The token-based replay approach described in \cite{Alessandro-token-based-replay-ATAED19} 
is used to annotate the \emph{places} and the performance on the \emph{arcs}.
This approach improves the approach \cite{anne_confcheck_is} and the implementation is considerably faster.
For each place, the number of \emph{produced} $p$, \emph{consumed} $c$, \emph{remaining} $r$, and \emph{missing} $m$ 
tokens are computed and displayed.
These values are obtained by ``playing the token game'' using the flattened event log $L^{\mi{ot}}$ and accepting Petri net $\mi{SN}^{\mi{ot}}$ for each object type $\mi{ot}$. This is possible because each place has precisely one type.
The numbers $p$ and $c$ refer to the number of produced and consumed tokens (reported per place).
The number of missing tokens $m$ refers to situations where a token is not present in the place 
although the log suggests that the output transition has fired. 
The number of remaining tokens $r$ refers to the tokens that remain after replaying the event log.
Our token-based replay approach is able to deal with silent transitions and duplicate transitions (i.e., the labeling function $l$ is partial or non-injective).
See \cite{anne_confcheck_is,Alessandro-token-based-replay-ATAED19} for details.

For performance-related annotations, 
the sets of delays based on differences between the production times of tokens and the consumption times of tokens are used.
Based on these measurements, minimum, maximum, average, variance, etc.\ can be calculated.

The annotations related to the transitions are derived directly from the event log (i.e., without replaying the event log).
This way we can add the frequencies of transitions, the average number of objects involved, and the number of unique objects to the model.

\subsection{Functionalities of the Tool}

Figure~\ref{fig:screenshotsWebInterface} shows two screenshots of our \emph{PM4Py-MDL} tool. 
The following functions are supported:
\begin{itemize}
\item \emph{Importing} and \emph{exporting} of object-centric event logs in different formats. 
The currently supported formats are Multi-Dimensional Logs (MDL), Parquet and XOC (format connected to OCBC models).
\item A range of \emph{object-centric process discovery approaches} are supported. 
There are also several target formats next to the object-centric Petri nets introduced in this paper. 
Also Multiple ViewPoint (MVP) models are supported. These are essentially Directly Follows Graphs \cite{centeris-keynote2019} with colored arcs, see \cite{Alessandro-StarStar-SIMPDA-CEUR2018,Alessandro-MVP-SIMPDA-ext2019}.
The approach presented in the paper can be combined with different low-level discovery techniques.
In the examples, we use the Inductive Miner Directly-Follows process discovery algorithm \cite{sander-scalable-procmin-SOSYM}.
\item It is possible to set various \emph{thresholds} to influence the discovery process, e.g., the minimal number of occurrences for activities and paths. It is also possible to specify, for each object type, the activities that are considered for that type.
\item Several methods to \emph{explore} the raw event data are provided (e.g., statistics on the number of related objects per type and distribution of events over time). These annotations can be attached to places, transitions, and arcs.
\item Token-based replay is supported for \emph{performance} and \emph{conformance} analysis. This allows for the identification of bottlenecks and deviating behavior. 
\item It is possible to \emph{filter} based on activities, paths, number of related objects per type. Also, timeframe and attribute-based filters are supported.
\item There is support for \emph{clustering} and event \emph{correlation} based on event graphs.
\item At any point in time, it is possible to \emph{flatten} an object-centric event log onto a \emph{traditional event log} by selecting an object type. The resulting event log can be analyzed using conventional process mining techniques.
\end{itemize}
\begin{figure*}[!t]
\centering
\begin{minipage}{0.45\textwidth}
\resizebox{\columnwidth}{!}{%
\includegraphics{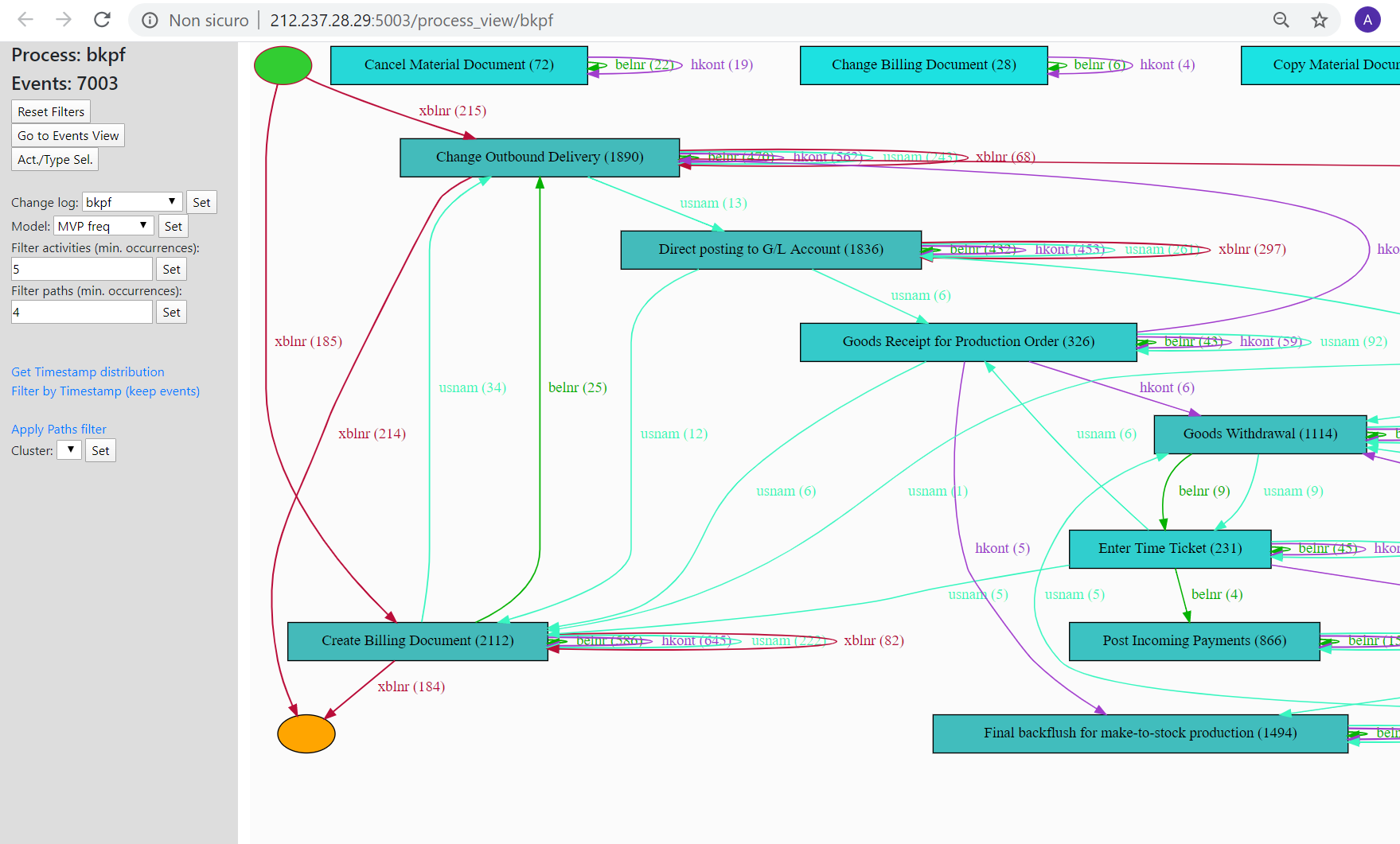}
}
\end{minipage}
\begin{minipage}{0.45\textwidth}
\resizebox{\columnwidth}{!}{%
\includegraphics{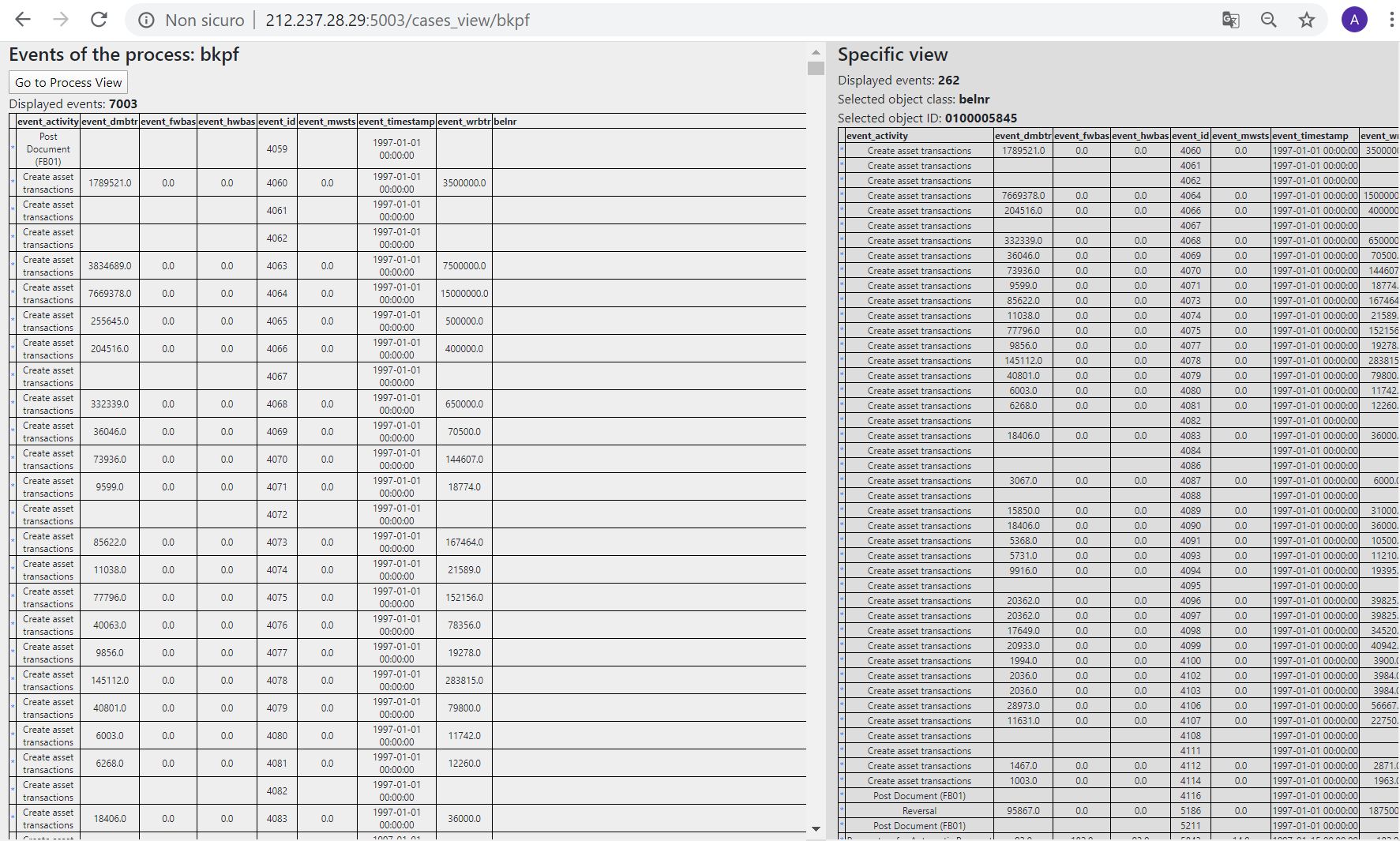}
}
\end{minipage}
\caption{Web interface that is supporting the functionalities offered by the \emph{PM4Py-MDL} library. The main components are the process discovery (left) and event exploration (right) ones.}
\label{fig:screenshotsWebInterface}
\end{figure*}

The web interface is organized mainly in two different components: process discovery and event exploration (see Figure~\ref{fig:screenshotsWebInterface}).
The visualization is highly interactive.
The nodes are clickable in such a way that the statistics about the events of such activity can be inspected and filtering options can be set.
The event exploration shows the events of the log in an interactive way.
It is possible to interact with the related objects and show all the events related to an object in another panel. This way the understanding the lifecycle of objects is facilitated.
Next, we evaluate the scalability of the approach and the implementation.

\subsection{Scalability of the Approach and Implementation}

The aim of this subsection is to analyze the scalability of the discovery of 
object-centric Petri nets as implemented in the \emph{PM4Py-MDL} tool.
We expect the discovery of object-centric Petri nets to be scalable, 
because the steps that are involved have at most linear complexity,
excluding the application of the process discovery algorithm 
on the flattened logs. 
Moreover, we also support discovery techniques that are linear in the event log (given a bounded number of activities).

To analyze scalability, we use variants of the ``running-example'' object-centric event log also used in other parts of the paper. Three different settings have been examined:
\begin{enumerate}
\item The execution time of the algorithm in terms of \emph{the number of events} in the event log (while keeping the number of unique activities and the number of objects per event constant).
\item The execution time of the algorithm in terms of \emph{the number of unique activities} in the event log (while keeping the number of events and the number of objects per event constant).
\item The execution time of the algorithm in terms of \emph{the number of objects per event} (while keeping the number of unique activities and the number of events constant).
\end{enumerate}
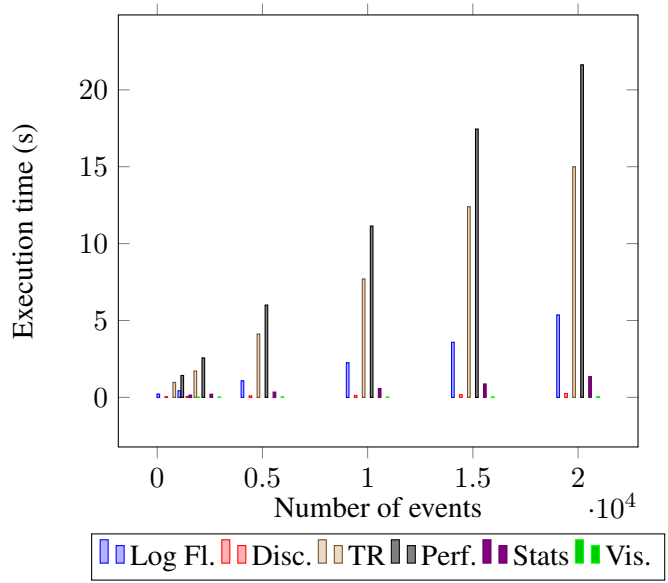
\begin{figure*}[!b]
\centering
\begin{tikzpicture}
\begin{axis}[
ylabel=Execution time (s),
xlabel=Number of events,
enlargelimits=0.15,
legend style={at={(0.5,-0.20)},
anchor=north,legend columns=-1},
ybar,
bar width=1pt,
]
\legend{Log Fl.,Disc.,TR,Perf.,Stats,Vis.}
\addplot
    coordinates { (1000.0, 0.217)  (2000.0, 0.431)  (5000.0, 1.078)  (10000.0, 2.255)  (15000.0, 3.582)  (20000.0, 5.361) };

\addplot
    coordinates { (1000.0, 0.036)  (2000.0, 0.048)  (5000.0, 0.095)  (10000.0, 0.125)  (15000.0, 0.172)  (20000.0, 0.251) };

\addplot
    coordinates { (1000.0, 0.973)  (2000.0, 1.714)  (5000.0, 4.121)  (10000.0, 7.701)  (15000.0, 12.399)  (20000.0, 15.003) };

\addplot
    coordinates { (1000.0, 1.424)  (2000.0, 2.566)  (5000.0, 6.005)  (10000.0, 11.143)  (15000.0, 17.452)  (20000.0, 21.63) };

\addplot
    coordinates { (1000.0, 0.147)  (2000.0, 0.205)  (5000.0, 0.344)  (10000.0, 0.578)  (15000.0, 0.873)  (20000.0, 1.355) };

\addplot
    coordinates { (1000.0, 0.016)  (2000.0, 0.016)  (5000.0, 0.031)  (10000.0, 0.016)  (15000.0, 0.031)  (20000.0, 0.047) };
\end{axis}
\end{tikzpicture}
\caption{Detailed analysis of the overall execution time of the approach when increasing the number of events of the log. The measurements are grouped for the five sublogs. The columns inside a group represent
event log flattening (\textbf{Log Fl.}), discovery (\textbf{Disc.}), token-based replay (\textbf{TR}), computing performance annotations (\textbf{Perf.}), computing statistics (\textbf{Stats}), and  visualization (\textbf{Vis.}).}
\label{fig:dissection}
\end{figure*}

\subsubsection{Increasing the Number of Events}
\label{sec:increasingNumberEvents}

Figures~\ref{fig:dissection} and \ref{fig:execTimeAct}(a) show the relationship between the overall execution time and the number of events in the object-centric event log.
Figure \ref{fig:execTimeAct}(a) shows a linear relationship between the number of events and the execution time.
Our initial ``running-example'' log contains 22,367 events. Different subsets of different sizes are taken 
such that the set of unique activities remains constant (we just consider fewer orders). In other words, the process is observed over shorter time periods. Analyzing the whole log takes less than a minute.
This may seem long for a relatively small event log. However, the time needed for discovery is less than a second.
Figure~\ref{fig:dissection} splits the analysis time into six different components:
\begin{itemize}
\item The time needed for the log flattening operations for all event logs (\textbf{Log Fl.}).
\item The time needed for the process discovery operations (\textbf{Disc.}). In these experiments, we use the inductive miner.
\item The time needed for the token-based replay operations (\textbf{TR}).
\item The time needed for computing the performance annotations based on the results of the token-based replay (\textbf{Perf.}).
\item The time needed for the calculation of additional statistics from the log (\textbf{Stats}).
\item The time needed for the visualization (\textbf{Vis.}).
\end{itemize}
Figure~\ref{fig:dissection} clearly shows that most time is spent on the token-based replay operations (\textbf{TR})
and the computation of the performance annotations (\textbf{Perf.}).
The first is done per control-flow variant (to avoid repeatedly solving the same problem) and the second one per object.
This explains why \textbf{Perf.} takes more time than \textbf{TR}.
Also, the event log preprocessing (\textbf{Log Fl.}) takes substantial time.
Interestingly, the discovery itself is very fast compared to the other components.
\pgfplotstableread[col sep=space,row sep=newline,header=true]{
x   y
1000 2.813
2000 4.98
5000 11.674
10000 21.818
15000 34.509
20000 43.647
}\ocpnNEvents
\pgfplotstableread[col sep=space,row sep=newline,header=true]{
x   y
1 5.26
2 6.74
3 7.02
4 9.25
5 9.85
6 10.342
7 10.9
8 12.18
9 13.257
10 13.564
11 13.581
}\ocpnNActivities
\pgfplotstableread[col sep=space,row sep=newline,header=true]{
x   y
1.0 11.06
1.45 15.240
2.12 17.418
3.24 25.89
4.35 30.782
5.47 35.56
6.59 40.314
7.71 42.933
8.82 43.281
9.22 44.532
}\ocpnNObjectsEvent
\begin{figure*}[!t]
\centering
\begin{minipage}{0.45\textwidth}
\resizebox{\columnwidth}{!}{%
\begin{tikzpicture}
  \begin{axis}[xlabel={Number of events}, ylabel={Execution time (s)}]
  \addplot[smooth] table {\ocpnNEvents};
\addplot[only marks,mark=*,mark options={color=red}] table {\ocpnNEvents};
\end{axis}
\end{tikzpicture}
} \\
(a) Execution time while increasing the number of events.
\end{minipage}
\begin{minipage}{0.45\textwidth}
\resizebox{\columnwidth}{!}{%
\begin{tikzpicture}
  \begin{axis}[xlabel={Number of activities}, ylabel={Execution time (s)}]
    \addplot[smooth] table {\ocpnNActivities};
\addplot[only marks,mark=*,mark options={color=blue}] table {\ocpnNActivities};
\end{axis}
\end{tikzpicture}
} \\
(b) Execution time while increasing the number of activities.
\end{minipage} \\
\vspace{5mm}
\begin{minipage}{0.45\textwidth}
\resizebox{\columnwidth}{!}{%
\begin{tikzpicture}
  \begin{axis}[xlabel={Objects per Event}, ylabel={Execution time (s)}]
    \addplot[smooth] table {\ocpnNObjectsEvent};
\addplot[only marks,mark=*,mark options={color=green}] table {\ocpnNObjectsEvent};
\end{axis}
\end{tikzpicture}
} \\
c) Execution time while increasing the number of objects per event.
\end{minipage}
\caption{Scalability assessment of the object-centric Petri nets discovery algorithm. The different graphs show the overall time (including replay and annotation) when varying of the number of events, the number of activities, and the number of objects per event.}
\label{fig:execTimeAct}
\end{figure*}

Figures~\ref{fig:dissection} and \ref{fig:execTimeAct}(a) show that the characteristics of our approach are similar to process mining on classical event logs. It takes more time to replay the event log to collect conformance and performance statistics than to discover the process model using techniques such as the inductive miner. 
This also holds for traditional process mining techniques using a single case notion.

\subsubsection{Increasing the Number of Activities}
\label{sec:increasingNumberActivities}

Figure~\ref{fig:execTimeAct}(b) shows the execution time when increasing the number of unique activities.
The event logs used were created using activity filtering while keeping the number of events constant.
Table~\ref{tab:tableNumberActivities} shows the number of activities, the number of events, and the overall time needed. In row $k$, the $k$ most frequent activities are retained and the event log is further filtered to have precisely 8159 events. 
The growth in overall computation time is explained by the fact that the most expensive operations are the token-based replay and computing the performance annotations, and the complexity of these operations grows linearly with the average
length of the trace.\footnote{Token-based replay, in contrast to other approaches such as alignments,
does not suffer from the increase of the size of the trace, since decisions are made locally.}
\begin{table}[ht]
\label{tab:tableNumberActivities}
\caption{The execution time while increasing the number of unique activities.}
\centering
\resizebox{0.50\columnwidth}{!}{%
\begin{tabular}{|c|c|c|}
\hline
{\bf Number of Activities} & {\bf Number of events} & {\bf Execution time} \\
\hline
1 & 8159 & 5.26 \\
2 & 8159 & 6.74 \\
3 & 8159 & 7.02 \\
4 & 8159 & 9.25 \\
5 & 8159 & 9.85 \\
6 & 8159 & 10.34 \\
7 & 8159 & 10.90 \\
8 & 8159 & 12.18 \\
9 & 8159 & 13.26 \\
10 & 8159 & 13.56 \\
11 & 8159 & 13.58 \\
\hline
\end{tabular}}
\end{table}

\subsubsection{Increasing the Number of Related Objects Per Event}

Figure~\ref{fig:execTimeAct}(c) shows the execution time when the number of related objects per event is increased.
To analyze such a setting, different subsets of the ``running-example'' event log were created in such a way that the number of events and the number of different activities does not change.
The set of related objects is selected such that 
at least one related object (of any type) remains for each event.
The linear relation is as expected, since events are replicated for each object during analysis.
Experiments also confirm that there is a linear relationship between the overall analysis time and the number of object \emph{types} (not shown).

Overall, the results are very encouraging. \emph{The approach scales linear in the number of events, the number of unique activities, and the number of objects.} The discovery times are negligible compared to the time needed for conformance checking and performance analysis. Hence, the approach can be applied to real-world event data.

\begin{table}[thb!]
\caption{Fragment of a larger object-centric event log with 22,367 events and five object types:
\emph{orders},
\emph{items},
\emph{products},
\emph{customers},
\emph{packages}.
There are 2000 different orders, 8159 items, 20 products, 17 customers, and 1325 packages.
The table shows a few sample events and the first three object types.}
\label{tab:mdlLogView}
\centering
\resizebox{0.99\columnwidth}{!}{%
\begin{tabular}{|l|l|l|l|l|}
\hline
    event\_activity &     event\_timestamp &                                             orders &                                              items &                                           products \\
\hline
       place order & 2019-05-20 09:07:47 &                                         ['990001'] &           ['880001', '880002', '880003', '880004'] &  ['Echo Show 8', 'Fire Stick 4K', 'Echo', 'Echo... \\
       place order & 2019-05-20 10:35:21 &                                         ['990002'] &           ['880005', '880006', '880007', '880008'] &      ['iPad', 'Kindle', 'iPad Air', 'MacBook Air'] \\
         pick item & 2019-05-20 10:38:17 &                                         ['990002'] &                                         ['880006'] &                                         ['Kindle'] \\
     confirm order & 2019-05-20 11:13:54 &                                         ['990001'] &           ['880001', '880002', '880003', '880004'] &  ['Echo Show 8', 'Fire Stick 4K', 'Echo', 'Echo... \\
         pick item & 2019-05-20 11:20:13 &                                         ['990001'] &                                         ['880002'] &                                  ['Fire Stick 4K'] \\
       place order & 2019-05-20 12:30:30 &                                         ['990003'] &           ['880009', '880010', '880011', '880012'] &  ['iPad Air', 'iPhone 11', 'Fire Stick', 'iPhon... \\
     confirm order & 2019-05-20 12:34:16 &                                         ['990003'] &           ['880009', '880010', '880011', '880012'] &  ['iPad Air', 'iPhone 11', 'Fire Stick', 'iPhon... \\
 item out of stock & 2019-05-20 13:54:37 &                                         ['990001'] &                                         ['880004'] &                                    ['Echo Studio'] \\
       place order & 2019-05-20 14:20:47 &                                         ['990004'] &                               ['880013', '880014'] &                     ['Echo Studio', 'Echo Show 8'] \\
 item out of stock & 2019-05-20 15:19:49 &                                         ['990003'] &                                         ['880009'] &                                       ['iPad Air'] \\
       place order & 2019-05-20 16:01:22 &                                         ['990005'] &                               ['880015', '880016'] &                           ['iPad Pro', 'iPad Air'] \\
         pick item & 2019-05-20 16:56:02 &                                         ['990004'] &                                         ['880014'] &                                    ['Echo Show 8'] \\
         pick item & 2019-05-20 17:08:25 &                                         ['990002'] &                                         ['880008'] &                                    ['MacBook Air'] \\
       place order & 2019-05-20 17:22:31 &                                         ['990006'] &                     ['880017', '880018', '880019'] &       ['Echo Show 8', 'Fire Stick 4K', 'iPhone X'] \\
         pick item & 2019-05-20 17:51:15 &                                         ['990003'] &                                         ['880011'] &                                     ['Fire Stick'] \\
         pick item & 2019-05-20 18:15:00 &                                         ['990002'] &                                         ['880007'] &                                       ['iPad Air'] \\
     confirm order & 2019-05-20 18:36:37 &                                         ['990004'] &                               ['880013', '880014'] &                     ['Echo Studio', 'Echo Show 8'] \\
       place order & 2019-05-20 19:04:49 &                                         ['990007'] &                     ['880020', '880021', '880022'] &   ['Echo Show 8', 'Echo Dot', 'Kindle Paperwhite'] \\
$\ldots$ & $\ldots$ & $\ldots$ & $\ldots$ & $\ldots$ \\
 \hline
\end{tabular}
}
\end{table}

\section{Example Application}
\label{sec:appl}

To illustrate the feasibility of the approach and corresponding \emph{PM4Py-MDL} implementation, we use the larger example briefly mentioned in the introduction (see Figure~\ref{f-intro}).
The object-centric event log in CSV format can be obtained from \url{https://github.com/Javert899/pm4py-mdl/blob/master/example_logs/mdl/mdl-running-example.mdl}.
A small fragment of the log, showing three selected object types, is visualized in Table~\ref{tab:mdlLogView}. It can be considered to be an extension of the smaller examples used before.
In total, there are 22,367 events. There are five object types:
\emph{orders},
\emph{items},
\emph{products},
\emph{customers},
\emph{packages}.
The event log contains information about 2000 different orders, 8159 items, 20 products, 17 customers, and 1325 packages.
Hence, the average number of items in one order is 4.08 and the average number of items in one package is 6.16.

We can filter out specific ``activity - object type'' $(a,\mi{ot})$ combinations. This corresponds to removing objects related to activity $a$ and object type $\mi{ot}$.
In Table~\ref{tab:mdlLogView}, we removed all objects related to \emph{customers} and \emph{packages} for all activities.
This boils down to removing the columns with customer and package information.
We can also remove the rows related to certain activities. However, we can also use more fine-grained filtering where we keep specific ``activity - object type'' combinations.
\begin{table}
\caption{The first two columns show the ``activity - object type'' combinations used for analysis. 
For example, \emph{place order} events also had information about products and customers, but these object types were removed.
\emph{failed delivery} events also had information about orders, products, and customers, but these were removed. Etc.
The last three columns show statistics for the ``activity - object type'' combinations in the original event log (only for the object types \emph{orders}, \emph{items}, and \emph{packages}).
The three values are reported: the minimum number of objects / the average number of objects / the maximum number of objects.}
\label{tab:mdlObjTypes}
\centering
\resizebox{0.8\columnwidth}{!}{
\begin{tabular}{|l|c||c|c|c|}
\hline
Activity & Retained object types & Orders per event & Items per event & Packages per event \\
\hline
place order & orders, items & 1~/~1.00~/~1 & 1~/~4.08~/~15 & 0~/~0.00~/~0 \\
confirm order & orders, items & 1~/~1.00~/~1 & 1~/~4.08~/~15 & 0~/~0.00~/~0 \\
item out of stock & items & 1~/~1.00~/~1 & 1~/~1.00~/~1 & 0~/~0.00~/~0 \\
reorder item & items & 1~/~1.00~/~1 & 1~/~1.00~/~1 & 0~/~0.00~/~0 \\
pick item & items & 0~/~0.00~/~0 & 1~/~1.00~/~1 & 0~/~0.00~/~0 \\
payment reminder & orders & 1~/~1.00~/~1 & 1~/~4.18~/~14 & 0~/~0.00~/~0 \\
pay order & orders & 1~/~1.00~/~1 & 1~/~4.08~/~15 & 0~/~0.00~/~0 \\
create package & items, packages & 1~/~3.32~/~9 & 1~/~6.16~/~22 & 1~/~1.00~/~1 \\
send package & items, packages & 1~/~3.32~/~9 & 1~/~6.16~/~22 & 1~/~1.00~/~1 \\
failed delivery & items, packages & 1~/~3.21~/~8 & 1~/~5.95~/~18 & 1~/~1.00~/~1 \\
package delivered & items, packages & 1~/~3.31~/~9 & 1~/~6.16~/~22 & 1~/~1.00~/~1 \\
\hline
\end{tabular}
}
\end{table}

Table~\ref{tab:mdlObjTypes} shows example statistics for the 22,367 events in the original object-centric event log.
For each activity, the minimum number of objects, the average number of objects, and the maximum number of objects of a given type are indicated.
For example, \emph{place order} events always refer to precisely one order object and a variable number of item objects (minimum=1, average=4.08, maximum=15) and
\emph{send package} events always refer to precisely one package object,  a variable number of item objects (minimum=1, average=6.16, maximum=22),
 and a variable number of order objects (minimum=1, average=3.32, maximum=9).
\begin{figure}[htbp]
\centering
\fbox{\includegraphics[height=0.9\textheight]{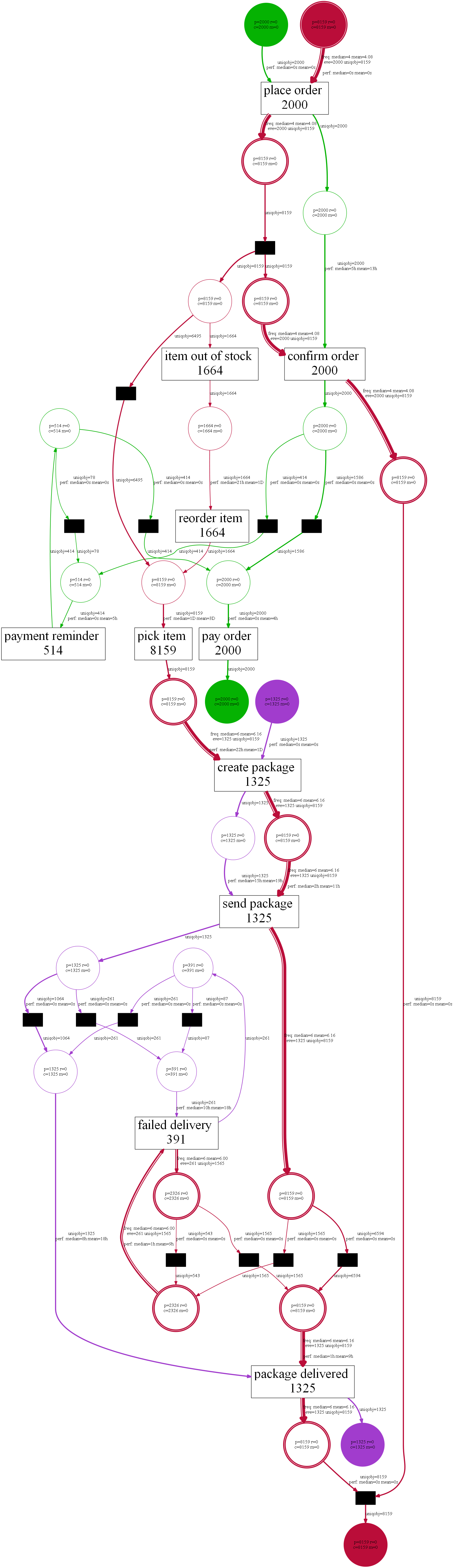}}
\caption{Object-centric Petri net discovered based on the example log considering three object types:  \emph{orders} (green), \emph{items} (red), and \emph{packages} (violet).}
\label{fig:petri1}
\end{figure}

For our running example, we considered the ``activity - object type'' combinations depicted in Table~\ref{tab:mdlObjTypes},
i.e., we retain object types \emph{orders}, \emph{items}, and \emph{packages}, keep all activities, but remove less relevant object types for some of the activities.
Starting from the event log in Table~\ref{tab:mdlLogView} and the ``activity - object type'' combinations in Table~\ref{tab:mdlObjTypes},
our discovery approach returns the object-centric Petri net shown in Figure~\ref{fig:petri1}.

The overall figure is hardly readable. However, we can use the filtering approaches discussed and seamlessly simplify the model (e.g., removing infrequent activities and selecting fewer object types).
Moreover, we can zoom in on the different aspects of the model.
\begin{figure}[htbp]
\centering
\fbox{\includegraphics[width=0.40\textwidth]{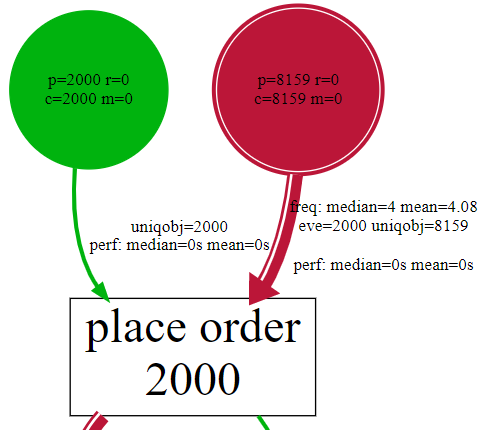}}
\caption{Fragment of the model showing the \emph{place order} activity. There are 2000 unique orders in the log and \emph{place order} occurs for each of them once.
There are 8159 unique items distributed over the 2000 orders. The diagnostics show that, on average, 4.08 item objects are consumed from the red place of type \emph{items}.}
\label{fig:orderPlacement}
\end{figure}
\begin{figure}[htbp]
\centering
\fbox{\includegraphics[width=0.40\textwidth]{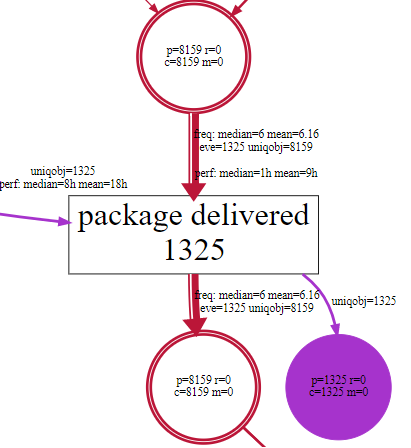}}
\caption{Fragment of the model showing the \emph{package delivered} activity. This activity corresponds to the successful delivery of packages composed of multiple items.
There are 8159 unique items distributed over 1325 packages. All packages were delivered as reflected by the frequency of \emph{package delivered}.
The mean number of item objects consumed and produced by the transition is 6.16.
The number of package objects consumed and produced by the transition is always 1. Also, the average times are reported.}
\label{fig:packageDelivery}
\end{figure}
\begin{figure}[htbp]
\centering
\fbox{\includegraphics[width=0.40\textwidth]{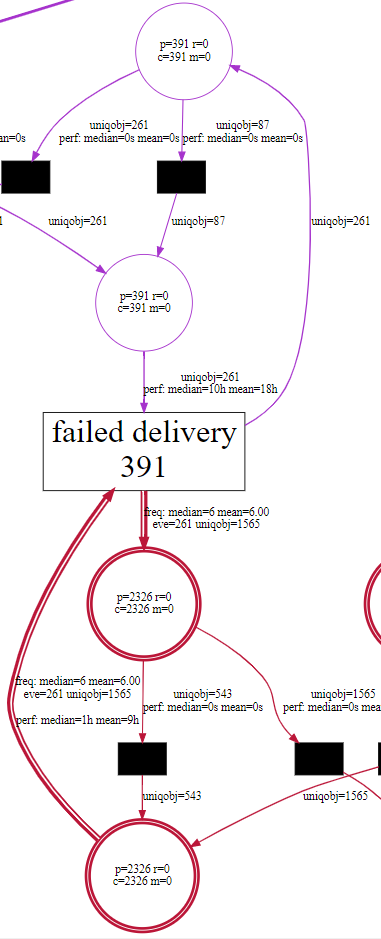}}
\caption{Fragment of the model showing the \emph{failed delivery} activity.
There were 391 failed deliveries, 261 packages had at least one failed delivery,
and 1565 items out of the 8159 where involved in at least one failed delivery.
87 packages containing 543 items had a failed delivery two or more times.}
\label{fig:failedDelivery}
\end{figure}

Figure~\ref{fig:orderPlacement} shows a fragment of the larger object-centric Petri net in Figure~\ref{fig:petri1}.
The green source place is of type \emph{orders}.
The red source place is of type \emph{items}.
Activity \emph{place order} occurred 2000 times, consuming precisely one token from the green place and a variable number of tokens from the red place.
The compound double arrow reflects this, and the inscription shows that on average 4.08 item objects were consumed.

Figure~\ref{fig:packageDelivery} shows another fragment. The {\it package delivered} activity is the final activity of the life-cycle of both packages and items.
The two compound double arrows denote that variable numbers of item objects are consumed and produced.
The mean number of item objects consumed and produced by {\it package delivered} is 6.16.
The annotations on the arcs tell that the average time from the previous activity for packages to this activity is 18 hours.
The average time from the previous activity for items to this activity is 9 hours.

Figure~\ref{fig:failedDelivery} zooms in on the failed deliveries. There were 391 failed deliveries involving 261 packages (87 failed multiple times) and 1565 items.

All the places in Figure~\ref{fig:petri1} also show replay information: $p$ is the number of tokens produced for the place,
$c$ is the number of tokens consumed, $m$ is the number of tokens missing, $r$ is the number of tokens remaining (see Section~\ref{subsec:impl}).
In this example, the model fits perfectly. Therefore, $m=0$ and $r=0$ for all places.

From the discovered object-centric Petri net, we can also generate simpler views, create traditional event logs, and deploy traditional process mining techniques for further analysis.
Moreover, the example shows many insights that could not have been discovered using traditional approaches.
By flattening the event log, the relations between different types of objects would be lost.
Moreover, any attempt to look at different types of objects would result in non-existing loops and misleading frequency/time diagnostics.

This section highlighted the main advantages of using our approach over traditional models using 
a separate process model for each object type. To summarize:
\begin{itemize}
\item Figure~\ref{fig:petri1} (and the corresponding model fragments) provides an overview of the \emph{whole} process and the \emph{interactions} between the different object types. When considering each object type as a separate case notion, we get multiple \emph{disconnected} models that do \emph{not} show these interactions.
\item \emph{Deficiency}, \emph{convergence}, and \emph{divergence} problems are avoided (cf.\ Definition~\ref{def:fltproblems}). All events are taken into count precisely once, i.e., events do not disappear and are not replicated unintentionally. Moreover, artificial loops due to divergence are avoided.
\item Using token-based replay, we are able to project \emph{performance and conformance information} onto one overall model. Most of the statistics would not be visible in the flattened process models (e.g., the average number of objects involved in an activity).
\end{itemize}

\section{Related Work}
\label{sec:rw}

This section discusses 
traditional process mining techniques using a single case notion,
modeling approaches dealing with multiple object types, and
process mining approaches dealing with multiple object types.

\subsection{Traditional Process Mining Techniques Using a Single Case Notion}

In the introduction, we mentioned several process discovery approaches based on 
classical event logs using a single case notion \cite{process-mining-book-2016}.
Many of these techniques discover classical Petri nets (e.g., place transition nets), e.g., region-based approaches can be used to derive places
\cite{Badouel-Darondeau-book-regions-2015,BaDa98,lorenz_BPM2007,DBLP:journals/fuin/BergenthumDLM08,DBLP:journals/fuin/BergenthumDML09,DBLP:conf/bpm/CarmonaCK08,DBLP:journals/tc/CarmonaCK10,DBLP:conf/apn/CarmonaCKKLY08,Cortadella98,zero-safe-nets-synthesis,boudewijn_runs_ToPNoC-special-issue-PN2011,Jetty-TCS-2012,Jetty-TCS-2017,Lorenz_ACSD_07,models-from-scenarios-ToPNoC,lorenz_atpn_2006,Lorenz_WSC_07,carmona-PN2010,bas-ilp-computing}.
The region-based process discovery techniques are just a subset of all approaches to derive process models from event logs.
The inductive mining techniques \cite{sander-infreq-bpi2013-lnbip2014,sander-scalable-BPMDS2015}
and the so-called split miner \cite{split-miner} are examples of the state-of-the-art techniques to learn process models.
Commercial systems tend to use the Directly Follows Graph (DFG) having the obvious limitations explained in \cite{centeris-keynote2019}.
All of the above approaches assume a single case notion.
This is consistent with traditional process models ranging from \emph{workflow nets} \cite{aaljcsc,soundness-FACS} and \emph{process trees} \cite{sander-tree-disc-PN2013} to \emph{Business Process Modeling Notation (BPMN) models} \cite{BPMN-OMG-2-formal} and \emph{Event-driven Process Chains (EPCs)} \cite{scheer94} which assume a single case notion.

\subsection{Modeling Techniques Using Multiple Object Types}

Although most process models use a single case notion, 
the problem that many processes cannot be captured properly in this way was identified early on.
IBM's \emph{FlowMark} system already supported the so-called ``bundle'' concept to handle cases composed of subcases \cite{flowmarkmanual}.
This is related to the \emph{multiple instance patterns}, i.e., a category of \emph{workflow patterns} identified around the turn of the century \cite{aal_patterns_dapd}.
One of the first process modeling notations trying to address the problem were the so-called \emph{proclets} \cite{aalprocletscoopis,aalprocletsjcis}.
Proclets are lightweight interacting workflow processes. By promoting interactions to first-class citizens, it is possible to model complex workflows in a more natural manner using proclets.

This was followed by other approaches such as the \emph{artifact-centric modeling notations} \cite{Artifact-Centric_BPM2007,CohnH:2009:business_artifacts,compliance-artifacts-bpm2011,NigamC:2003:artifacts}.
See \cite{dirk-artifact-PN2019} for an up-to-date overview of the challenges that arise when instances of processes may interact with each other in a one-to-many
or many-to-many fashion.

\subsection{Process Mining Techniques Using Multiple Object Types}

Most of the work done on interacting processes with converging and diverging instances has focused on developing novel modeling notations and supporting the implementation of such processes.
Only a few approaches focused on the problem in a process mining context. This is surprising since one quickly encounters the problem when applying process mining to ERP systems from SAP, Oracle, Microsoft, and other vendors of enterprise software. This problem was also raised in Section~5.5 of \cite{process-mining-book-2016} which discusses the need to ``flatten'' event data to produce traditional process models.

In \cite{eduardo-BPMDS2016} techniques are described to extract ``non-flat'' event data from source systems and prepare these for traditional process mining.
The \emph{eXtensible Event Stream} (XES) format \cite{XES-standard-2013} is the official IEEE standard for storing event data and supported by many process mining vendors.
XES requires a case notion to correlate events.
Next to the standard IEEE XES format \cite{XES-standard-2013}, new storage formats such as  \emph{eXtensible Object-Centric} (XOC) \cite{Extract-Object-Centric-CF-2018} have been proposed
to deal with object-centric data (e.g., database tables) having one-to-many and many-to-many relations.
The XOC format does not require a case notion to avoid flattening multi-dimensional data. An XOC log can precisely store the evolution of the database along with corresponding events.
 An obvious drawback is that XOC logs tend to be very large.

The approaches described in \cite{BIS-artifactconformance-lnbip2011,zeus-artifact-2011,Xixi-TSC-2015} focus on interacting processes where each process uses its own case identifiers.
In \cite{Xixi-TSC-2015} interacting artifacts are discovered from ERP systems.
In \cite{BIS-artifactconformance-lnbip2011} traditional conformance checking was adapted to check compliance for interacting artifacts.

One of the main challenges is that artifact models tend to become complex and difficult to understand. In an attempt to tackle this problem,
Van Eck et al.\ use a simpler setting with multiple perspectives, each modeled by a simple transition system \cite{DBLP:conf/wecwis/EckSA17,Multi-instance-Mining-BPM-WS-2018}.
These are also called \emph{artifact-centric process models} but are simpler than the models used in  \cite{Artifact-Centric_BPM2007,CohnH:2009:business_artifacts,BIS-artifactconformance-lnbip2011,zeus-artifact-2011,compliance-artifacts-bpm2011,NigamC:2003:artifacts,Xixi-TSC-2015}.
The state of a case is decomposed onto one state per perspective, thus simplifying the overall model.
Relations between sub-states are viewed as correlations rather than explicit causality constraints.
Concurrency only exists between the different perspectives and not within an individual perspective.
In a recent extension, each perspective can be instantiated multiple times, i.e., many-to-many relations between artifact types can be visualized \cite{Multi-instance-Mining-BPM-WS-2018}.

The above techniques have the drawback that the overall process is not visualized in a single diagram, but shown as a collection of interconnected diagrams using different (sub-)case notions.
The so-called \emph{Object-Centric Behavioral Constraint} (OCBC) models address this problem and also incorporate the data perspective in a single diagram \cite{DBLP:conf/dlog/AalstAMT17,DBLP:journals/corr/AalstLM17,OCBC-festschrift-Guarino,BIS-OCBCdisc-lnbip2017}.
OCBC models extend data models with a behavioral perspective. Data models can easily deal with many-to-many and one-to-many relationships. This is exploited to create process models that can also model complex interactions between different types of instances. Classical multiple-instance problems are circumvented by using the data model for event correlation.
Activities are related to the data perspective and have ordering constraints inspired by declarative languages like \emph{Declare} \cite{declareCSRD09}.
Instead of LTL-based constraints, simpler cardinality constraints are used.
Several discovery techniques have been developed for OCBC models \cite{BIS-OCBCdisc-lnbip2017}.
It is also possible to check conformance and project performance information on such models.
OCBC models are appealing because they faithfully describe the relationship between behavior and data and are able to capture all information in a single integrated diagram.
However, OCBC models tend to be too complex and the corresponding discovery and conformance checking techniques are not very scalable.

The complexity and scalability problems of OCBC models led to the development of the so-called \emph{Multiple ViewPoint (MVP) models}, earlier named StarStar models \cite{Alessandro-StarStar-SIMPDA-CEUR2018,Alessandro-MVP-SIMPDA-ext2019}.
MVP models are learned from data stored in relational databases. Based on the relations and timestamps in a traditional database, first, a so-called E2O graph is built that relates events and objects.
Based on the E2O graph, an E2E multigraph is learned that relates events through objects. Finally, an A2A multigraph is learned to relate activities.
The A2A graph shows relations between activities and each relation is based on one of the object classes used as input.
This is a very promising approach because it is simple and scalable. The approach to discover object-centric Petri nets can be seen as a continuation of the work in \cite{Alessandro-StarStar-SIMPDA-CEUR2018,Alessandro-MVP-SIMPDA-ext2019}.

Although commercial vendors have recognized the problems related to convergence and divergence of event data, there is no real support for concepts
comparable to artifact-centric models, Object-Centric Behavioral Constraint (OCBC) models, and Multiple ViewPoint (MVP) models.
Yet, there are a few initial attempts implemented in commercial systems.
An example is \emph{Celonis}, which supports the use of a secondary case identifier to avoid ``Spaghetti-like'' models where concurrency between sub-instances is translated into loops.
The directly-follows graphs in Celonis do not consider interactions between sub-instances, thus producing simpler models.
Another example is the multi-level discovery technique supported by \emph{myInvenio}.
The resulting models can be seen as simplified MVP models where different activities may correspond to different case notions (but one case notion per activity).
The problem of this approach is that, in reality, the same event may refer to multiple case notions and choosing one is often misleading, especially since it influences the frequencies shown in the diagram.

In spite of the recent progress in process mining, problems related to multiple interacting process instances  have not been solved adequately.
One of the problems is the lack of standardized event data that goes beyond the ``flattened'' event data found in XES.
Hence, process mining competitions tend to focus on classical event logs.
In earlier papers \cite{wvda-keynote-SEFM2019,Alessandro-StarStar-SIMPDA-CEUR2018,Alessandro-MVP-SIMPDA-ext2019}, we already stressed the need for object-centric process mining.
In this paper, we provided a concrete, but at the same time generic, discovery approach to learning object-centric Petri nets from object-centric events logs.

\section{Conclusion}
\label{sec:concl}

When looking at data models or database schemas, there are often one-to-many and many-to-many relations between different types of objects relevant for a process.
Since mainstream process modeling and process mining approaches enforce the use of a specific case notion, the modeler or analyst is forced to select a specific perspective.
This problem can be partly addressed by extracting multiple event logs to cover the different case notions and considering 
one model per case notion.
It would be better to have one, more holistic, process model that is showing the interactions between the different types of objects.
Moreover, the need to pick one or more specific case notions for analysis leads to the divergence and convergence problems discussed in this paper.

Therefore, this paper uses \emph{object-centric event logs} as a representation in between the actual data in the information system and traditional event logs (e.g., based on XES).
Object-centric event logs do not depend on a case notion. Instead, events may refer to arbitrary sets of objects.
One event may refer to multiple objects of different types. Next to using a different input format, we also use a different target language: \emph{object-centric Petri nets}.
These nets are a restricted variant of colored Petri nets where places are typed, tokens refer to objects, and transitions correspond to activities.
Unlike other mainstream notations, a transition can consume and produce a variable number of objects of different types.
We presented a concrete, but also generic, approach to discover object-centric Petri nets from object-centric event logs.
The approach has been implemented in \emph{PM4Py} and various applications show that the approach provides novel insights and is highly scalable (linear in the number of objects, object types, events, and activities).
Therefore, the ideas are directly implementable in commercial tools and 
the existing software can be used to analyze real-life event data in larger organizations.

This is the first paper that aims to learn object-centric Petri nets from object-centric event logs. Our findings show lots of opportunities for further research.
These include:
\begin{itemize}
  \item We aim to develop conformance checking techniques based on object-centric Petri nets and object-centric event logs. Next to checking whether the event log can be replayed on the process  model, it is interesting to detect outliers using the cardinalities. In the current implementation, we already report missing and remaining tokens, but these are based on the flattened event logs.
  \item The approach presented is generic and can embed different process discovery algorithms independently working on flattened events logs (inductive miner, region-based techniques, etc.). The results are then folded into object-centric Petri nets.  It is interesting to compare the different approaches and develop more integrated approaches (e.g., first discover a process model for one object type and then iteratively add the other object types). Moreover, it would be good to have dedicated quality measures (e.g., complexity and precision).
  \item Object-centric Petri nets in their current form can be seen as ``over-approximations'' of the actual behavior. It is interesting to think of ways to make the model more precise (e.g., automatically detecting guards or relating splits and joins). For example, in Figure~\ref{f-example-pm-3}, transition \emph{marked as complete} should join the same set of objects earlier involved in an occurrence of transition \emph{place order}. Similarly, we would like to add stochastics to the model (e.g., a probability distribution for the number of items in an order).
   \item The current object-centric event logs only contain object identifiers and no properties of objects. If an object identifier refers to a patient or customer, 
  we do not know her age, weight, address, income, etc. If an object identifier refers to an order or machine, we do not know its value, The absence of object attributes automatically leads to the ``over-approximations'' mentioned. Hence, we are developing extended object-centric event logs.
  \item We also plan to investigate more sophisticated forms of performance analysis that go beyond adding timing a frequency diagnostics to transition, places, and arcs. How do the different object types influence each other? Next to analyzing the interactions between objects, we would like to better support the link to CPN Tools for ``what if'' analysis (e.g., replaying the event log on an improved process).
  \item We also aim to create a comprehensive, publicly available, set of object-centric event logs.
\end{itemize}

~\\
{\bf Acknowledgments}: We thank the Alexander von Humboldt (AvH) Stiftung for supporting our research.

\bibliographystyle{plain}
\bibliography{lit}

\begin{thebibliography}{10}

\bibitem{aaljcsc}
{W.M.P. van der} Aalst.
\newblock {The Application of Petri Nets to Workflow Management}.
\newblock {\em The Journal of Circuits, Systems and Computers}, 8(1):21--66,
  1998.

\bibitem{process-mining-book-2016}
{W.M.P. van der} Aalst.
\newblock {\em {Process Mining: Data Science in Action}}.
\newblock Springer-Verlag, Berlin, 2016.

\bibitem{centeris-keynote2019}
{W.M.P. van der} Aalst.
\newblock {A Practitioner's Guide to Process Mining: Limitations of the
  Directly-Follows Graph}.
\newblock In {\em {International Conference on Enterprise Information Systems
  (Centeris 2019)}}, volume 164 of {\em Procedia Computer Science}, pages
  321--328. Elsevier, 2019.

\bibitem{wvda-keynote-SEFM2019}
{W.M.P. van der} Aalst.
\newblock {Object-Centric Process Mining: Dealing With Divergence and
  Convergence in Event Data}.
\newblock In P.C. {\"{O}}lveczky and G.~Sala{\"{u}}n, editors, {\em Software
  Engineering and Formal Methods (SEFM 2019)}, volume 11724 of {\em Lecture
  Notes in Computer Science}, pages 3--25. Springer-Verlag, Berlin, 2019.

\bibitem{DBLP:conf/dlog/AalstAMT17}
{W.M.P. van der} Aalst, A.~Artale, M.~Montali, and S.~Tritini.
\newblock {Object-Centric Behavioral Constraints: Integrating Data and
  Declarative Process Modelling}.
\newblock In {\em Proceedings of the 30th International Workshop on Description
  Logics (DL 2017)}, volume 1879 of {\em {CEUR} Workshop Proceedings}.
  CEUR-WS.org, 2017.

\bibitem{aalprocletscoopis}
{W.M.P. van der} Aalst, P.~Barthelmess, C.A. Ellis, and J.~Wainer.
\newblock {Workflow Modeling using Proclets}.
\newblock In O.~Etzion and P.~Scheuermann, editors, {\em {7th International
  Conference on Cooperative Information Systems (CoopIS 2000)}}, volume 1901 of
  {\em Lecture Notes in Computer Science}, pages 198--209. Springer-Verlag,
  Berlin, 2000.

\bibitem{aalprocletsjcis}
{W.M.P. van der} Aalst, P.~Barthelmess, C.A. Ellis, and J.~Wainer.
\newblock {Proclets: A Framework for Lightweight Interacting Workflow
  Processes}.
\newblock {\em International Journal of Cooperative Information Systems},
  10(4):443--482, 2001.

\bibitem{soundness-FACS}
{W.M.P. van der} Aalst, {K.M. van} Hee, {A.H.M. ter} Hofstede, N.~Sidorova,
  H.M.W. Verbeek, M.~Voorhoeve, and M.T. Wynn.
\newblock {Soundness of Workflow Nets: Classification, Decidability, and
  Analysis}.
\newblock {\em Formal Aspects of Computing}, 23(3):333--363, 2011.

\bibitem{aal_patterns_dapd}
{W.M.P. van der} Aalst, {A.H.M. ter} Hofstede, B.~Kiepuszewski, and A.P.
  Barros.
\newblock {Workflow Patterns}.
\newblock {\em Distributed and Parallel Databases}, 14(1):5--51, 2003.

\bibitem{DBLP:journals/corr/AalstLM17}
{W.M.P. van der} Aalst, G.~Li, and M.~Montali.
\newblock {Object-Centric Behavioral Constraints}.
\newblock {\em CoRR}, abs/1703.05740, 2017.

\bibitem{declareCSRD09}
{W.M.P. van der} Aalst, M.~Pesic, and H.~Schonenberg.
\newblock {Declarative Workflows: Balancing Between Flexibility and Support}.
\newblock {\em Computer Science - Research and Development}, 23(2):99--113,
  2009.

\bibitem{mbp-aal-stahl-2011}
{W.M.P. van der} Aalst and C.~Stahl.
\newblock {\em {Modeling Business Processes: A Petri Net Oriented Approach}}.
\newblock MIT Press, Cambridge, MA, 2011.

\bibitem{OCBC-festschrift-Guarino}
A.~Artale, D.~Calvanese, M.~Montali, and {W.M.P. van der} Aalst.
\newblock {Enriching Data Models with Behavioral Constraints}.
\newblock In S.~Borgo, editor, {\em {Ontology Makes Sense (Essays in honor of
  Nicola Guarino)}}, pages 257--277. IOS Press, 2019.

\bibitem{split-miner}
A.~Augusto, R.~Conforti, M.~Marlon, M.~{La Rosa}, and A.~Polyvyanyy.
\newblock {Split Miner: Automated Discovery of Accurate and Simple Business
  Process Models from Event Logs}.
\newblock {\em Knowledge Information Systems}, 59(2):251--284, May 2019.

\bibitem{Badouel-Darondeau-book-regions-2015}
E.~Badouel, L.~Bernardinello, and P.~Darondeau.
\newblock {\em {Petri Net Synthesis}}.
\newblock Texts in Theoretical Computer Science. An {EATCS} Series.
  Springer-Verlag, Berlin, 2015.

\bibitem{BaDa98}
E.~Badouel and P.~Darondeau.
\newblock {Theory of Regions}.
\newblock In W.~Reisig and G.~Rozenberg, editors, {\em Lectures on Petri Nets
  I: Basic Models}, volume 1491 of {\em Lecture Notes in Computer Science},
  pages 529--586. Springer-Verlag, Berlin, 1998.

\bibitem{lorenz_BPM2007}
R.~Bergenthum, J.~Desel, R.~Lorenz, and S.~Mauser.
\newblock {Process Mining Based on Regions of Languages}.
\newblock In G.~Alonso, P.~Dadam, and M.~Rosemann, editors, {\em International
  Conference on Business Process Management (BPM 2007)}, volume 4714 of {\em
  Lecture Notes in Computer Science}, pages 375--383. Springer-Verlag, Berlin,
  2007.

\bibitem{DBLP:journals/fuin/BergenthumDLM08}
R.~Bergenthum, J.~Desel, R.~Lorenz, and S.~Mauser.
\newblock {Synthesis of Petri Nets from Finite Partial Languages}.
\newblock {\em Fundamenta Informaticae}, 88(4):437--468, 2008.

\bibitem{DBLP:conf/apn/BergenthumDLM08}
R.~Bergenthum, J.~Desel, R.~Lorenz, and S.~Mauser.
\newblock {Synthesis of Petri Nets from Scenarios with VipTool}.
\newblock In {\em Applications and Theory of Petri Nets (Petri Nets 2008)},
  volume 5062 of {\em Lecture Notes in Computer Science}, pages 388--398.
  Springer-Verlag, Berlin, 2008.

\bibitem{DBLP:journals/fuin/BergenthumDML09}
R.~Bergenthum, J.~Desel, S.~Mauser, and R.~Lorenz.
\newblock {Synthesis of Petri Nets from Term Based Representations of Infinite
  Partial Languages}.
\newblock {\em Fundamenta Informaticae}, 95(1):187--217, 2009.

\bibitem{Alessandro-StarStar-SIMPDA-CEUR2018}
A.~Berti and {W.M.P. van der} Aalst.
\newblock {StarStar Models: Using Events at Database Level for Process
  Analysis}.
\newblock In P.~Ceravolo, {M. van} Keulen, and M.T.~Gomez Lopez, editors, {\em
  {International Symposium on Data-driven Process Discovery and Analysis
  (SIMPDA 2018)}}, volume 2270 of {\em CEUR Workshop Proceedings}, pages
  60--64. CEUR-WS.org, 2018.

\bibitem{Alessandro-token-based-replay-ATAED19}
A.~Berti and {W.M.P. van der} Aalst.
\newblock {Reviving Token-based Replay: Increasing Speed While Improving
  Diagnostics}.
\newblock In {\em Proceedings of the International Workshop on Algorithms and
  Theories for the Analysis of Event Data (ATAED 2019)}, volume 2371 of {\em
  {CEUR} Workshop Proceedings}, pages 87--103. CEUR-WS.org, 2019.

\bibitem{Alessandro-MVP-SIMPDA-ext2019}
A.~Berti and {W.M.P. van der} Aalst.
\newblock {Discovering Multiple Viewpoint Models from Relational Databases}.
\newblock In P.~Ceravolo, {M. van} Keulen, and M.T.~Gomez Lopez, editors, {\em
  {Postproceedings International Symposium on Data-driven Process Discovery and
  Analysis}}, volume 379 of {\em Lecture Notes in Business Information
  Processing}, pages 24--51. Springer-Verlag, Berlin, 2020.

\bibitem{Artifact-Centric_BPM2007}
K.~Bhattacharya, C.~Gerede, R.~Hull, R.~Liu, and J.~Su.
\newblock {Towards Formal Analysis of Artifact-Centric Business Process
  Models}.
\newblock In G.~Alonso, P.~Dadam, and M.~Rosemann, editors, {\em International
  Conference on Business Process Management (BPM 2007)}, volume 4714 of {\em
  Lecture Notes in Computer Science}, pages 288--304. Springer-Verlag, Berlin,
  2007.

\bibitem{DBLP:conf/bpm/CarmonaCK08}
J.~Carmona, J.~Cortadella, and M.~Kishinevsky.
\newblock {A Region-Based Algorithm for Discovering Petri Nets from Event
  Logs}.
\newblock In {\em Business Process Management (BPM 2008)}, pages 358--373,
  2008.

\bibitem{DBLP:journals/tc/CarmonaCK10}
J.~Carmona, J.~Cortadella, and M.~Kishinevsky.
\newblock {New Region-Based Algorithms for Deriving Bounded Petri Nets}.
\newblock {\em IEEE Transactions on Computers}, 59(3):371--384, 2010.

\bibitem{DBLP:conf/apn/CarmonaCKKLY08}
J.~Carmona, J.~Cortadella, M.~Kishinevsky, A.~Kondratyev, L.~Lavagno, and
  A.~Yakovlev.
\newblock {A Symbolic Algorithm for the Synthesis of Bounded Petri Nets}.
\newblock In {\em Applications and Theory of Petri Nets (Petri Nets 2008)},
  pages 92--111, 2008.

\bibitem{CohnH:2009:business_artifacts}
D.~Cohn and R.~Hull.
\newblock {Business Artifacts: A Data-centric Approach to Modeling Business
  Operations and Processes}.
\newblock {\em IEEE Data Engineering Bulletin}, 32(3):3--9, 2009.

\bibitem{Cortadella98}
J.~Cortadella, M.~Kishinevsky, L.~Lavagno, and A.~Yakovlev.
\newblock {Deriving {Petri} Nets from Finite Transition Systems}.
\newblock {\em IEEE Transactions on Computers}, 47(8):859--882, August 1998.

\bibitem{zero-safe-nets-synthesis}
P.~Darondeau.
\newblock {On the Synthesis of Zero-Safe Nets}.
\newblock In {\em Concurrency, Graphs and Models}, volume 5065 of {\em Lecture
  Notes in Computer Science}, pages 364--378. Springer-Verlag, Berlin, 2008.

\bibitem{eduardo-BPMDS2016}
E.~Gonz{\'{a}}lez~L{\'{o}}pez de~Murillas, H.A. Reijers, and {W.M.P. van der}
  Aalst.
\newblock {Connecting Databases with Process Mining: A Meta Model and Toolset}.
\newblock In R.~Schmidt, W.~Guedria, I.~Bider, and S.~Guerreiro, editors, {\em
  Enterprise, Business-Process and Information Systems Modeling (BPMDS 2015)},
  volume 248 of {\em Lecture Notes in Business Information Processing}, pages
  231--249. Springer-Verlag, Berlin, 2016.

\bibitem{DeRe96}
J.~Desel and W.~Reisig.
\newblock {The Synthesis Problem of Petri Nets}.
\newblock {\em Acta Informatica}, 33(4):297--315, 1996.

\bibitem{boudewijn_runs_ToPNoC-special-issue-PN2011}
{B.F. van} Dongen, J.~Desel, and {W.M.P. van der} Aalst.
\newblock {Aggregating Causal Runs into Workflow Nets}.
\newblock In K.~Jensen, {W.M.P. van der} Aalst, M.~Ajmone Marsan,
  G.~Franceschinis, J.~Kleijn, and L.M. Kristensen, editors, {\em Transactions
  on Petri Nets and Other Models of Concurrency (ToPNoC VI)}, volume 7400 of
  {\em Lecture Notes in Computer Science}, pages 334--363. Springer-Verlag,
  Berlin, 2012.

\bibitem{DBLP:conf/wecwis/EckSA17}
{M.L. van} Eck, N.~Sidorova, and {W.M.P. van der} Aalst.
\newblock {Guided Interaction Exploration in Artifact-centric Process Models}.
\newblock In {\em IEEE Conference on Business Informatics (CBI 2017)}, pages
  109--118. {IEEE} Computer Society, 2017.

\bibitem{ehrenfeucht_regions}
A.~Ehrenfeucht and G.~Rozenberg.
\newblock {Partial (Set) 2-Structures - Part 1 and Part 2}.
\newblock {\em Acta Informatica}, 27(4):315--368, 1989.

\bibitem{dirk-artifact-PN2019}
D.~Fahland.
\newblock {Describing Behavior of Processes with Many-to-Many Interactions}.
\newblock In S.~Donatelli and S.~Haar, editors, {\em {Applications and Theory
  of Petri Nets 2019}}, volume 11522 of {\em Lecture Notes in Computer
  Science}, pages 3--24. Springer-Verlag, Berlin, 2019.

\bibitem{BIS-artifactconformance-lnbip2011}
D.~Fahland, {M. De} Leoni, {B. van} Dongen, and {W.M.P. van der} Aalst.
\newblock {Behavioral Conformance of Artifact-Centric Process Models}.
\newblock In A.~Abramowicz, editor, {\em Business Information Systems (BIS
  2011)}, volume~87 of {\em Lecture Notes in Business Information Processing},
  pages 37--49. Springer-Verlag, Berlin, 2011.

\bibitem{zeus-artifact-2011}
D.~Fahland, {M. De} Leoni, {B. van} Dongen, and {W.M.P. van der} Aalst.
\newblock {Many-to-Many: Some Observations on Interactions in Artifact
  Choreographies}.
\newblock In D.~Eichhorn, A.~Koschmider, and H.~Zhang, editors, {\em
  Proceedings of the 3rd Central-European Workshop on Services and their
  Composition (ZEUS 2011)}, CEUR Workshop Proceedings, pages 9--15.
  CEUR-WS.org, 2011.

\bibitem{flowmarkmanual}
IBM.
\newblock {\em IBM MQSeries Workflow - Getting Started With Buildtime}.
\newblock IBM Deutschland Entwicklung GmbH, Boeblingen, Germany, 1999.

\bibitem{XES-standard-2013}
{IEEE Task Force on Process Mining}.
\newblock {XES Standard Definition}.
\newblock www.xes-standard.org, 2013.

\bibitem{cpnbook-jensen-2009}
K.~Jensen and L.M. Kristensen.
\newblock {\em {Coloured Petri Nets}}.
\newblock Springer-Verlag, Berlin, 2009.

\bibitem{Jetty-TCS-2012}
J.~Kleijn, M.Koutny, and M.~Pietkiewicz-Koutny.
\newblock {Regions of Petri nets with a/sync connections}.
\newblock {\em Theoretical Computer Science}, 454:189--198, 2012.

\bibitem{Jetty-TCS-2017}
J.~Kleijn, M.Koutny, M.~Pietkiewicz-Koutny, and G.~Rozenberg.
\newblock {Applying Regions}.
\newblock {\em Theoretical Computer Science}, 658:205--215, 2017.

\bibitem{sander-tree-disc-PN2013}
S.J.J. Leemans, D.~Fahland, and {W.M.P. van der} Aalst.
\newblock {Discovering Block-structured Process Models from Event Logs: A
  Constructive Approach}.
\newblock In J.M. Colom and J.~Desel, editors, {\em {Applications and Theory of
  Petri Nets 2013}}, volume 7927 of {\em Lecture Notes in Computer Science},
  pages 311--329. Springer-Verlag, Berlin, 2013.

\bibitem{sander-infreq-bpi2013-lnbip2014}
S.J.J. Leemans, D.~Fahland, and {W.M.P. van der} Aalst.
\newblock {Discovering Block-Structured Process Models from Event Logs
  Containing Infrequent Behaviour}.
\newblock In N.~Lohmann, M.~Song, and P.~Wohed, editors, {\em {Business Process
  Management Workshops, International Workshop on Business Process Intelligence
  (BPI 2013)}}, volume 171 of {\em Lecture Notes in Business Information
  Processing}, pages 66--78. Springer-Verlag, Berlin, 2014.

\bibitem{sander-scalable-BPMDS2015}
S.J.J. Leemans, D.~Fahland, and {W.M.P. van der} Aalst.
\newblock {Scalable Process Discovery with Guarantees}.
\newblock In K.~Gaaloul, R.~Schmidt, S.~Nurcan, S.~Guerreiro, and Q.~Ma,
  editors, {\em Enterprise, Business-Process and Information Systems Modeling
  (BPMDS 2015)}, volume 214 of {\em Lecture Notes in Business Information
  Processing}, pages 85--101. Springer-Verlag, Berlin, 2015.

\bibitem{sander-scalable-procmin-SOSYM}
S.J.J. Leemans, D.~Fahland, and {W.M.P. van der} Aalst.
\newblock {Scalable Process Discovery and Conformance Checking}.
\newblock {\em Software and Systems Modeling}, 17(2):599--631, 2018.

\bibitem{Extract-Object-Centric-CF-2018}
G.~Li, E.~Gonz{\'{a}}lez~L{\'{o}}pez de~Murillas, R.~Medeiros de~Carvalho, and
  {W.M.P. van der} Aalst.
\newblock {Extracting Object-Centric Event Logs to Support Process Mining on
  Databases}.
\newblock In J.~Mendling and H.~Mouratidis, editors, {\em {Information Systems
  in the Big Data Era, CAiSE Forum 2018}}, volume 317 of {\em Lecture Notes in
  Business Information Processing}, pages 182--199. Springer-Verlag, Berlin,
  2018.

\bibitem{BIS-OCBCdisc-lnbip2017}
G.~Li, R.~{Medeiros de Carvalho}, and {W.M.P. van der} Aalst.
\newblock {Automatic Discovery of Object-Centric Behavioral Constraint Models}.
\newblock In W.~Abramowicz, editor, {\em Business Information Systems (BIS
  2017)}, volume 288 of {\em Lecture Notes in Business Information Processing},
  pages 43--58. Springer-Verlag, Berlin, 2017.

\bibitem{compliance-artifacts-bpm2011}
N.~Lohmann.
\newblock {Compliance by Design for Artifact-Centric Business Processes}.
\newblock In S.~Rinderle, F.~Toumani, and K.~Wolf, editors, {\em {Business
  Process Management (BPM 2011)}}, volume 6896 of {\em Lecture Notes in
  Computer Science}, pages 99--115. Springer-Verlag, Berlin, 2011.

\bibitem{Lorenz_ACSD_07}
R.~Lorenz, R.~Bergenthum, J.~Desel, and S.~Mauser.
\newblock {Synthesis of Petri Nets from Finite Partial Languages}.
\newblock In T.~Basten, G.~Juh{\'a}s, and S.K. Shukla, editors, {\em
  International Conference on Application of Concurrency to System Design (ACSD
  2007)}, pages 157--166. IEEE Computer Society, 2007.

\bibitem{models-from-scenarios-ToPNoC}
R.~Lorenz, J.~Desel, and G.~Juhas.
\newblock {Models from Scenarios}.
\newblock In K.~Jensen, {W.M.P. van der} Aalst, G.~Balbo, M.~Koutny, and
  K.~Wolf, editors, {\em Transactions on Petri Nets and Other Models of
  Concurrency (ToPNoC VII)}, volume 7480 of {\em Lecture Notes in Computer
  Science}, pages 314--371. Springer-Verlag, Berlin, 2013.

\bibitem{lorenz_atpn_2006}
R.~Lorenz and G.~Juhas.
\newblock {Towards Synthesis of Petri Nets from Scenarios}.
\newblock In S.~Donatelli and P.S. Thiagarajan, editors, {\em {Application and
  Theory of Petri Nets 2006}}, volume 4024 of {\em Lecture Notes in Computer
  Science}, pages 302--321. Springer-Verlag, Berlin, 2006.

\bibitem{Lorenz_WSC_07}
R.~Lorenz and G.~Juh{\'a}s.
\newblock {How to Synthesize Nets from Languages: A Survey}.
\newblock In S.G. Henderson, B.~Biller, M.~Hsieh, J.~Shortle, J.~D. Tew, and
  R.~R. Barton, editors, {\em Proceedings of the Wintersimulation Conference
  (WSC 2007)}, pages 637--647. IEEE Computer Society, 2007.

\bibitem{Xixi-Conf-Check-bpi2014-lnbip2015}
X.~Lu, D.~Fahland, and {W.M.P. van der} Aalst.
\newblock {Conformance Checking Based on Partially Ordered Event Data}.
\newblock In F.~Fournier and J.~Mendling, editors, {\em {Business Process
  Management Workshops, International Workshop on Business Process Intelligence
  (BPI 2014)}}, volume 202 of {\em Lecture Notes in Business Information
  Processing}, pages 75--88. Springer-Verlag, Berlin, 2015.

\bibitem{Xixi-TSC-2015}
X.~Lu, M.~Nagelkerke, {D. van de} Wiel, and D.~Fahland.
\newblock {Discovering Interacting Artifacts from ERP Systems}.
\newblock {\em IEEE Transactions on Services Computing}, 8(6):861--873, 2015.

\bibitem{NigamC:2003:artifacts}
A.~Nigam and N.S. Caswell.
\newblock {Business artifacts: An Approach to Operational Specification}.
\newblock {\em IBM Systems Journal}, 42(3):428--445, 2003.

\bibitem{BPMN-OMG-2-formal}
OMG.
\newblock {Business Process Model and Notation (BPMN)}.
\newblock Object Management Group, formal/2011-01-03, 2011.

\bibitem{anne_confcheck_is}
A.~Rozinat and {W.M.P. van der} Aalst.
\newblock {Conformance Checking of Processes Based on Monitoring Real
  Behavior}.
\newblock {\em Information Systems}, 33(1):64--95, 2008.

\bibitem{scheer94}
A.W. Scheer.
\newblock {\em {Business Process Engineering: Reference Models for Industrial
  Enterprises}}.
\newblock Springer-Verlag, Berlin, 1994.

\bibitem{carmona-PN2010}
M.~Sol{\'e} and J.~Carmona.
\newblock {Process Mining from a Basis of State Regions}.
\newblock In J.~Lilius and W.~Penczek, editors, {\em {Applications and Theory
  of Petri Nets 2010}}, volume 6128 of {\em Lecture Notes in Computer Science},
  pages 226--245. Springer-Verlag, Berlin, 2010.

\bibitem{Multi-instance-Mining-BPM-WS-2018}
M.L. van Eck, N.~Sidorova, and {W.M.P. van der} Aalst.
\newblock {Multi-instance Mining: Discovering Synchronisation in
  Artifact-Centric Processes}.
\newblock In F.~Daniel, Q.Z. Sheng, and H.~Motahari, editors, {\em Business
  Process Management Workshops, International Workshop on Business Process
  Intelligence (BPI 2018)}, volume 342 of {\em Lecture Notes in Business
  Information Processing}, pages 18--30. Springer-Verlag, Berlin, 2018.

\bibitem{bas-ilp-computing}
{S.J. van} Zelst, {B.F. van} Dongen, {W.M.P. van der} Aalst, and H.M.W Verbeek.
\newblock {Discovering Workflow Nets Using Integer Linear Programming}.
\newblock {\em Computing}, 100(5):529--556, 2018.

\end{thebibliography}

\end{document}